\documentclass[a4paper,UKenglish,cleveref, autoref]{lipics-v2019}

\title{When Maximum Stable Set can be solved in FPT time} 

\titlerunning{When Maximum Stable can be solved in FPT time}

\author{\'{E}douard Bonnet}{Univ Lyon, CNRS, ENS de Lyon, Université Claude Bernard Lyon 1, LIP UMR5668, France}{edouard.bonnet@ens-lyon.fr}{https://orcid.org/0000-0002-1653-5822}{}

\author{Nicolas Bousquet}{CNRS, G-SCOP laboratory, Grenoble-INP, France}{nicolas.bousquet@grenoble-inp.fr}{}{}

\author{St\'ephan Thomass\'e}{Univ Lyon, CNRS, ENS de Lyon, Université Claude Bernard Lyon 1, LIP UMR5668, France \and
Institut Universitaire de France}{stephan.thomasse@ens-lyon.fr}{}{}

\author{R\'emi Watrigant}{Univ Lyon, CNRS, ENS de Lyon, Université Claude Bernard Lyon 1, LIP UMR5668, France}{remi.watrigant@ens-lyon.fr}{https://orcid.org/0000-0002-6243-5910}{}

\authorrunning{\'E. Bonnet, N. Bousquet, S. Thomass\'e, R. Watrigant}

\Copyright{Édouard Bonnet, Nicolas Bousquet, Stéphan Thomassé, and Rémi Watrigant}

\ccsdesc[100]{Theory of computation → Graph algorithms analysis}
\ccsdesc[100]{Theory of computation → Fixed parameter tractability}

\keywords{Parameterized Algorithms, Independent Set, H-Free Graphs}

\category{}

\relatedversion{}

\supplement{}

\acknowledgements{}

\nolinenumbers 

\EventEditors{John Q. Open and Joan R. Access}
\EventNoEds{2}
\EventLongTitle{42nd Conference on Very Important Topics (CVIT 2016)}
\EventShortTitle{CVIT 2016}
\EventAcronym{CVIT}
\EventYear{2016}
\EventDate{December 24--27, 2016}
\EventLocation{Little Whinging, United Kingdom}
\EventLogo{}
\SeriesVolume{42}
\ArticleNo{23}

\usepackage{xspace}
\usepackage[table]{xcolor}
\usepackage{tikz}
\usepackage{array}
\usetikzlibrary{fit}
\usetikzlibrary{calc}
\usetikzlibrary{shapes}
\usetikzlibrary{decorations.pathreplacing}
\usepackage{algpseudocode,algorithm,algorithmicx}
\newcommand*\Let[2]{\State #1 $\gets$ #2}
\algrenewcommand\algorithmicrequire{\textbf{Precondition:}}
\algrenewcommand\algorithmicensure{\textbf{Postcondition:}}
\usepackage{relsize}
\usepackage{amsthm}

\makeatletter
\newcommand\footnoteref[1]{\protected@xdef\@thefnmark{\ref{#1}}\@footnotemark}
\makeatother

\newcommand{\indepset}{\textsc{Maximum Independent Set}\xspace}
\newcommand{\mis}{\textsc{MIS}\xspace}

\newenvironment{faugramsey}[1]{\textsc{#1-Ramsey-extracted Iterative Expansion MIS}\xspace}{}

\theoremstyle{plain}
\newtheorem{conjecture}[theorem]{Conjecture}
\newtheorem*{classicConj}{Classical MIS Dichotomy Conjecture(H)}
\newtheorem*{paramDichotomy}{Parameterized MIS Dichotomy(H)}
\newtheorem{observation}[theorem]{Observation}

\newcommand{\ie}{\textit{i.e.}\xspace}
\newcommand{\etal}{\textit{et al.}\xspace}
\newcommand{\wrt}{\textit{w.r.t.}\xspace}

\newcommand{\C}{\mathcal{C}}

\newcommand{\G}{\mathcal{G}}
\newcommand{\B}{\mathcal{B}}

\newcommand{\Ramsey}{Ram}

\newcommand{\pk}[1]{$P(#1)$\xspace}
\newcommand{\wc}{weakly connected\xspace}
\newcommand{\wcness}{weakly connectedness\xspace}

\begin{document}

\maketitle

\begin{abstract}
  \textsc{Maximum Independent Set} (\textsc{MIS} for short) is in general graphs the paradigmatic $W[1]$-hard problem.
  In stark contrast, polynomial-time algorithms are known when the inputs are restricted to structured graph classes such as, for instance, perfect graphs (which includes bipartite graphs, chordal graphs, co-graphs, etc.) or claw-free graphs.
  In this paper, we introduce some variants of co-graphs with \emph{parameterized noise}, that is, graphs that can be made into disjoint unions or complete sums by the removal of a certain number of vertices and the addition/deletion of a certain number of edges per incident vertex, both controlled by the parameter.
  We give a series of FPT Turing-reductions on these classes and use them to make some progress on the parameterized complexity of \textsc{MIS} in $H$-free graphs.
We show that for every fixed $t \geqslant 1$, \textsc{MIS} is FPT in $P(1,t,t,t)$-free graphs, where $P(1,t,t,t)$ is the graph obtained by substituting all the vertices of a four-vertex path but one end of the path by cliques of size $t$.
We also provide randomized FPT algorithms in dart-free graphs and in cricket-free graphs.
This settles the FPT/W[1]-hard dichotomy for five-vertex graphs $H$.
\end{abstract}

\section{Introduction}\label{sec:intro}
A \emph{stable set} or \emph{independent set} in a graph is a subset of vertices which are pairwise non-adjacent. 
Finding an independent set of maximum cardinality, called \indepset (or \mis for short), is not only NP-hard to solve \cite{GJ79} but also to approximate within ratio $n^{1-\varepsilon}$ \cite{Hastad96,Zuk07}.
One can then wonder if efficient algorithms exist with the additional guarantee that $k$, the size of the maximum stable set, is fairly small compared to $n$, the number of vertices of the input (think, for instance, $k \leqslant \log n$).
It turns out that, for any computable function $h = \omega(1)$ (but whose growth can be arbitrarily slow), \mis is unlikely to admit a polynomial-time algorithm even when $k \leqslant h(n)$.
In parameterized complexity terms, \mis is $W[1]$-hard \cite{DF13}.
More quantitatively, \mis cannot be solved in time $f(k)n^{o(k)}$ for any computable function $f$, unless the Exponential Time Hypothesis fails.
This is quite a statement when a trivial algorithm for \mis runs in time $n^{k+2}$, and a simple reduction to triangle detection yields a $n^{\frac{\omega k}{3}+O(1)}$-algorithm, where $\omega$ is the best exponent known for matrix multiplication.

It thus appears that \mis on general graphs is totally impenetrable.
This explains why efforts have been made on solving \mis in subclasses of graphs.
The most emblematic result in that line of works is a polynomial-time algorithm in perfect graphs \cite{Grotschel81}.
Indeed, perfect graphs generalize many graph classes for which \mis is in P: bipartite graphs, chordal graphs, co-graphs, etc.
In this paper, we put the focus on classes of graphs for which \mis can be solved in FPT time (rather than in polynomial-time). 
For graphs with bounded degree $\Delta$, the simple branching algorithm has FPT running time $(\Delta+1)^kn^{O(1)}$.
The same observation also works more generally for graphs with bounded average degree, or even $d$-degenerate graphs.
A non-trivial result is that \mis remains FPT in arguably the most general class of sparse graphs, \emph{nowhere dense graphs}.
Actually, deciding first-order formulas of size $k$ can be done in time $f(k)n^{1+\varepsilon}$ on any nowhere dense class of graphs \cite{GroheKS17}.
Since \mis and the complement problem, \textsc{Maximum Clique}, are both expressible by a first-order formula of length $O(k^2)$, $\exists v_1, \ldots, v_k \bigwedge_{i,j} (\neg)E(v_i,v_j)$, there is an FPT algorithm on nowhere dense graphs and also on complements of nowhere dense graphs.
A starting point of the present paper is to design FPT Turing-reductions in classes containing both very dense and very sparse graphs.

\subparagraph*{Co-graphs with parameterized noise.}
If $G$ and $H$ are two graphs, we can define two new graphs: $G \cup H$, their disjoint union, and $G \oplus H$ their (complete) sum, obtained from the disjoint union by adding all the edges from a vertex of $G$ to a vertex of $H$. 
Then, the hereditary class of co-graphs can be inductively defined by: $K_1$ (an isolated vertex) is a co-graph, and $G \cup H$ and $G \oplus H$ are co-graphs, if $G$ and $H$ are themselves co-graphs.
So the construction of a co-graph can be seen as a binary tree whose internal nodes are labeled by $\cup$ or $\oplus$, and leaves are $K_1$.
Finding the tree of operations building a given co-graph, sometimes called \emph{co-tree}, can be done in linear time \cite{CoPeSt85}.
This gives a simple algorithm to solve \mis on co-graphs: $\alpha(K_1)=1$, $\alpha(G \cup H)=\alpha(G)+\alpha(H)$, and $\alpha(G \oplus H)=\max(\alpha(G),\alpha(H))$.

We add a \emph{parameterized noise} to the notion of co-graphs.
More precisely, we consider graphs that can be made disjoint unions or complete sums by the deletion of $g(k)$ vertices and the edition (\ie, addition or deletion) of $d(k)$ edges per incident vertex.
We design a series of FPT Turing-reductions on several variants of these classes using the so-called \emph{iterative expansion} technique \cite{ChLiLuSzZh12,BonnetBCTW18}, Cauchy-Schwarz-like inequalities, and K\H{o}v\'ari-S\'os-Tur\'an's theorem.
This serves as a crucial foundation for the next part of the paper, where we explore the parameterized complexity of \mis in $H$-free graphs (\ie, graphs not containing $H$ as an induced subgraph).
However, we think that the FPT routines developed on \emph{co-graphs with parameterized noise} may also turn out to be useful outside the realm of $H$-free graphs.

\subparagraph*{Classical and parameterized dichotomies in $H$-free graphs.}
The question of whether \mis is in P or NP-complete in $H$-free graphs, for each connected graph $H$, goes back to the early eighties.
However, a full dichotomy is neither known nor does it seem within reach in the near future.
For three positive integers $i, j, k$, let $S_{i,j,k}$ be the tree with exactly one vertex of degree three, from which start three paths with $i$, $j$, and $k$ edges, respectively.
The claw is the graph $S_{1,1,1}$, thus $\{S_{i,j,k}\}_{i \leqslant j \leqslant k}$ is the set of all the subdivided claws.
We denote by $P_\ell$ the path on $\ell$ vertices.

If $G'$ is the graph obtained by subdividing each edge of a graph $G$ exactly $2t$ times, Alekseev observed that $\alpha(G') = \alpha(G) + t|E(G)|$ \cite{Ale82}.
This shows that \mis remains NP-hard on graphs which locally look like paths or subdivided claws (one can perform the subdivision on sub-cubic graphs $G$, on which \mis remains NP-complete).
In other words, if a connected graph $H$ is not a path nor a subdivided claw then \mis is NP-complete on $H$-free graphs \cite{Ale82}.
The \mis problem is easy on $P_4$-free graphs, which are exactly the co-graphs.
Already on $P_5$-free graphs, a polynomial algorithm is much more difficult to obtain.
This was done by Lokshtanov \etal \cite{LoVaVi14} using the framework of potential maximal cliques.
A quasi-polynomial algorithm was proposed for $P_6$-free graphs \cite{LokshtanovPL18}, and recently, a polynomial-time algorithm was found by Grzesik~\etal \cite{GrzesikKPP19}.
Brandst\"{a}dt and Mosca showed how to solve \mis in polynomial-time on ($P_7$, triangle)-free graphs \cite{BrandstadtM18}.
This result was then generalized by the same authors on ($S_{1,2,4}$, triangle)-free graphs \cite{abs-1806-09472}, and by Maffray and Pastor on ($P_7$, bull\footnote{the bull is obtained by adding a pendant neighbor to two distinct vertices of the triangle ($K_3$)})-free graphs (as well as ($S_{1,2,3}$, bull)-free graphs) \cite{MaffrayP18}.
Bacs\'o \etal \cite{Bacso19} presented a subexponential-time $2^{O(\sqrt{t n \log n})}$-algorithm in $P_t$-free graphs, for every integer $t$.
Nevertheless, the classical complexity of \mis remains wide open on $P_t$-free graphs, for $t \geqslant 7$.

On claw-free graphs \mis is known to be polynomial-time solvable \cite{Minty80,Sbihi80}.
Recently, this result was generalized to $\ell$claw-free graphs \cite{BrandstadtM18a} (where $\ell$claw is the disjoint union of $\ell$ claws).
On fork-free graphs (the fork is $S_{1,1,2}$) \mis admits a polynomial-time algorithm \cite{Alekseev04}, and so does its weighted variant \cite{Lozin08}.
The complexity of \mis is open for $S_{1,1,3}$-free graphs and $S_{1,2,2}$-free graphs, and there is no triple $i \leqslant j \leqslant k$, for which we know that \mis is NP-hard on $S_{i,j,k}$-free graphs.
Some subclasses of $S_{i,j,k}$-free graphs are known to admit polynomial algorithms for \mis: for instance $(S_{1,1,3},K_{t,t})$-free graphs \cite{DabrowskiLWZ16}, subcubic $S_{2,t,t}$-free graphs \cite{Harutyunyan18} (building upon \cite{LozinMR15}, and generalizing results presented in \cite{Malyshev13,Malyshev17} for subcubic \emph{planar} graphs), bounded-degree $tS_{1,t,t}$-free graphs \cite{LozinM07}, for any fixed positive integer $t$.
This leads to the following conjecture:
\begin{classicConj}
  For every connected graph $H$,

  \indepset in $H$-free graphs is in P iff $H \in \{P_\ell\}_\ell \cup \{S_{i,j,k}\}_{i \leqslant j \leqslant k}$. 
\end{classicConj}

An even stronger conjecture is postulated by Lozin (see Conjecture 1 in \cite{Lozin17}). 
Dabrowski \etal initiated a systematic study of the parameterized complexity of \mis on $H$-free graphs \cite{DabrowskiThesis,Dabrowski12}.
In a nutshell, parameterized complexity aims to design $f(k)n^{O(1)}$-algorithms (FPT algorithm, for Fixed-Parameter Tractable), where $n$ is the size of the input, and $k$ is the size of the solution (or another well-chosen parameter), for most often NP-hard problems.
The so-called $W$-hierarchy (and in particular, $W[1]$-hardness) and the Exponential Time Hypothesis (ETH) both provide a framework to rule out such a running time.
We refer the interested reader to two recent textbooks \cite{DF13,CyFoKoLoMaPiPiSa15} and to a survey on the ETH and its consequences \cite{surveyETH}.
In the language of parameterized complexity, the dichotomy problem is the following:
\begin{paramDichotomy}
  Is \mis (randomized) FPT or $W[1]$-hard in $H$-free graphs?  
\end{paramDichotomy}
This question may be even more challenging than its classical counterpart.
Indeed, there is no FPT algorithm known for the classical open cases: $P_7$-, $S_{1,1,3}$-, and $S_{1,2,2}$-free graphs.
Besides, the reduction of Alekseev \cite{Ale82} that we mentioned above does not show $W[1]$-hardness.
Thus, there are \emph{a priori} more candidates $H$ for which the parameterized status of \mis is open.
For instance, by Ramsey's theorem, \mis is FPT on $K_t$-free graphs, for any fixed $t$. Observe that a randomized FPT algorithm for a W[1]-hard problem is highly unlikely, as it would imply a randomized algorithm solving \textsc{$3$-SAT} in subexponential time. 

Dabrowski \etal showed that \mis is FPT\footnote{here and in what follows, the parameter is the size of the solution} in $H$-free graphs, for all $H$ on four vertices, except $H=C_4$ (the cycle on four vertices).
Thomassé \etal presented an FPT algorithm on bull-free graphs~\cite{ThTrVu17}, whose running time was later improved by Perret du Cray and Sau \cite{PedCrSa18}.
Bonnet \etal provided three variants of a parameterized counterpart of Alekseev's reduction \cite{BonnetBCTW18,BoBoChThWa18}.
Although the description of the open cases (see Figure~\ref{fig:FPTcandidates}) is not nearly as nice and compact as for the classical dichotomy, it is noteworthy that they almost correspond to paths and subdivided claws where vertices are blown-up into cliques.

\begin{figure}[h!]
  \centering
  \begin{tikzpicture}
    \begin{scope}[xshift=-6cm]
      \node[draw, circle] (a1) at (0,0) {} ;
      \node[draw, circle] (a2) at (0.75,0) {} ;
      \node[draw, circle] (a3) at (0.75,0.75) {} ;
      \node[draw, circle] (a4) at (0,0.75) {} ;
      \draw (a1) -- (a2) -- (a3) -- (a4) ;
      \draw[very thick, dotted] (a1) -- (a4) ;
    \end{scope}

      \begin{scope}[xshift=-4.6cm]
      \node[draw, circle] (c1) at (0.75,0.75) {} ;
      \node[draw, circle] (c2) at (0,0) {} ;
      \node[draw, circle] (c3) at (0.5,0) {} ;
      \node[draw, circle] (c4) at (1,0) {} ;
      \node[draw, circle] (c5) at (1.5,0) {} ;
      \draw (c3) -- (c1) -- (c2) ;
      \draw (c5) -- (c1) -- (c4) ;
      \end{scope}

    \begin{scope}[xshift=-2.5cm]
      \node[draw, circle] (b1) at (0,0) {} ;
      \node[draw, circle,fill = black] (b2) at (0,0.5) {} ;
      \node[draw, circle] (b3) at (0,1) {} ;
      \node[draw, circle] (b4) at (0.75,0) {} ;
      \node[draw, circle,fill = black] (b5) at (0.75,0.5) {} ;
      \node[draw, circle] (b6) at (0.75,1) {} ;
      \draw (b1) -- (b2) -- (b3) ;
      \draw (b4) -- (b5) -- (b6) ;
      \draw (b2) -- (b5) ;
    \end{scope}

    \draw[very thick,red] (-6.5,-0.3) -- (-1.2,1.3) ;
    \draw[very thick,red] (-1.2,-0.3) -- (-6.5,1.3) ;

    \node at (-0.75,0.45) {And} ;
    \node at (3.2,0.45) {Or} ;
    \draw [decorate,decoration={brace,amplitude=10pt}] (0,-1.6) -- (0,2.4) ;
    \draw [decorate,decoration={brace,amplitude=10pt}] (3.8,-1.5) -- (3.8,2.3) ;
    \draw [decorate,decoration={brace,amplitude=10pt}] (6.2,2.3) --  (6.2,-1.5) ;
    \draw [decorate,decoration={brace,amplitude=10pt}] (6.5,2.4) --  (6.5,-1.6) ;
    
    \def\t{0.22}
    \begin{scope}[xshift=0.2cm,yshift=0.45cm]
      \node[draw,rectangle,rounded corners,inner sep=\t cm] (v1) at (0,0) {} ;
      \node[draw,rectangle,rounded corners,inner sep=\t cm] (v2) at (0.75,0) {} ;
      \node[draw,rectangle,rounded corners,inner sep=\t cm] (v3) at (1.5,0) {} ;
      \node[draw,rectangle,rounded corners,inner sep=\t cm] (v4) at (2.25,0) {} ;
      \draw[very thick,red] (v1) -- (v2) ;
      \draw[very thick,blue] (v2) -- (v3) ;
      \draw[very thick,red] (v3) -- (v4) ;
    \end{scope}

      \begin{scope}[xshift=0cm,yshift=1.5cm]
      \node[draw,rectangle,rounded corners,inner sep=\t cm] (w1) at (5,0.5) {} ;
      \node[draw,rectangle,rounded corners,inner sep=\t cm] (w2) at (4.5,-0.25) {} ;
      \node[draw,rectangle,rounded corners,inner sep=\t cm] (w3) at (5,-0.25) {} ;
      \node[draw,rectangle,rounded corners,inner sep=\t cm] (w4) at (5.5,-0.25) {} ;

      \draw[very thick,red] (w2) -- (w1) -- (w3) ;
      \draw[very thick,red] (w1) -- (w4) ;
      \end{scope}

      \node at (5,0.76) {And} ;
      \begin{scope}[xshift=0cm,yshift=-1cm]
      \node[draw,rectangle,rounded corners,inner sep=\t cm] (z1) at (5,0.5) {} ;
      \node[draw,rectangle,rounded corners,inner sep=\t cm] (z2) at (4.6,-0.25) {} ;
      \node[draw,rectangle,rounded corners,inner sep=\t cm] (z3) at (5.4,-0.25) {} ;

      \node[draw,rectangle,rounded corners,inner sep=\t cm] (z0) at (5,1.2) {} ;
      \node[draw,rectangle,rounded corners,inner sep=\t cm] (z4) at (4,-0.25) {} ;
      \node[draw,rectangle,rounded corners,inner sep=\t cm] (z5) at (6,-0.25) {} ;
      \draw[very thick,red] (z1) -- (z2) -- (z3) -- (z1) ;
      \draw[very thick,red] (z0) -- (z1) ;
      \draw[very thick,red] (z2) -- (z4) ;
      \draw[very thick,red] (z3) -- (z5) ;
       \end{scope}
\end{tikzpicture}
  \caption{The dotted edge represents a path with at least one edge. The filled vertices emphasize two vertices with degree at least three in a tree. The rounded boxes are cliques. A red edge corresponds to a complete bipartite minus at most one edge. A blue edge correspond to a $2K_2$-free bipartite graph. The FPT connected candidates $H$ have to be chordal, without induced $K_{1,4}$ or trees with two branching vertices (\ie, vertices of degree at least $3$), and have to fit on a path with at most one blue edge (and the rest of red edges) or both in a subdivided claw and a line-graph of a subdivided claw with red edges only. A further restriction in the line-graph of subdivided claw is that three vertices each in a different clique of the triangle of red edges cannot induce a $K_1 \cup K_2$ (see \cite{BonnetBCTW18}).}
\label{fig:FPTcandidates}
\end{figure}
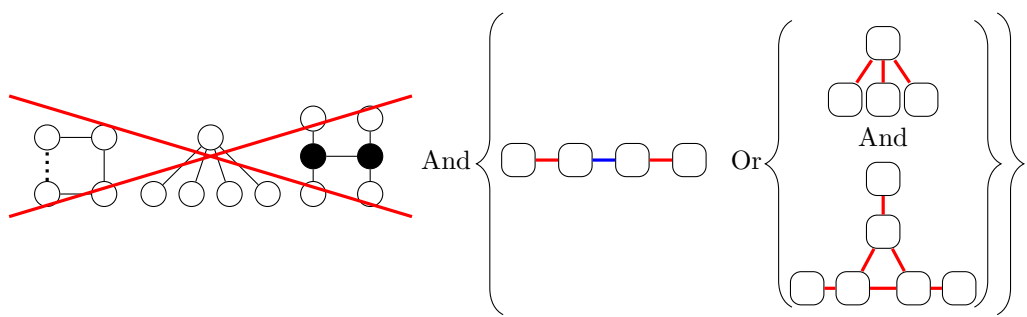

Let us make that idea more formal.
Substituting a graph $H$ at a vertex $v$ of a graph $G$ gives a new graph with vertex set $(V(G) \setminus \{v\}) \cup V(H)$, and the same edges as in $G$ and $H$, plus all the edges $xy$ where $x \in V(H)$, $y \in V(G)$, and $vy \in E(G)$. 
For a sequence of positive integers $a_1,a_2,\ldots,a_\ell$, we denote by \pk{a_1, a_2,\ldots, a_\ell} the graph obtained by substituting a clique $K_{a_i}$ at the $i$-th vertex of a path $P_\ell$, for every $i \in [\ell]$.
We also denote by $T(a,b,c)$ the graph obtained by substituting a clique $K_a$, $K_b$, and $K_c$ to the first, second, and third leaves, respectively, of a claw.
Thus, $T(1,1,1)$ is the claw and $T(1,1,2)$ is called the cricket (see Figure~\ref{fig:cricket}).

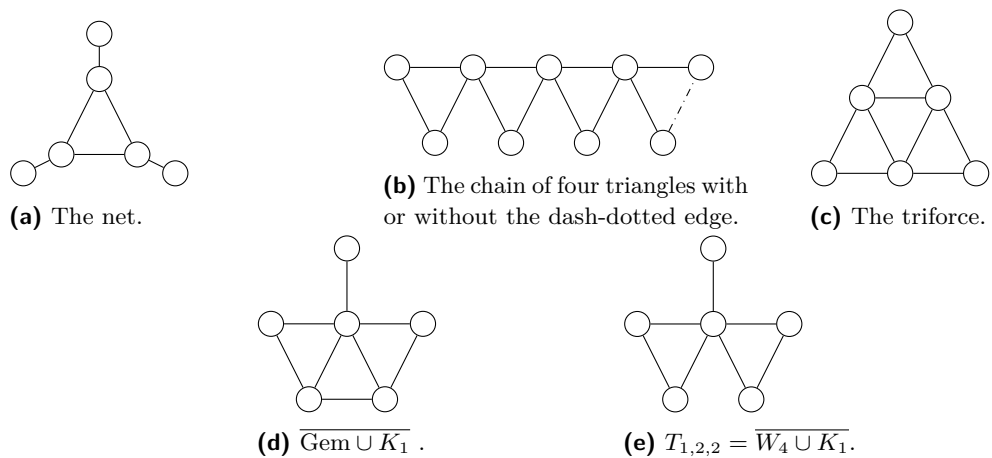
\begin{figure}
  \centering
   \begin{subfigure}[b]{0.3\textwidth}
  \begin{tikzpicture}
    \node[draw,circle] (a) at (0,-0.15) {} ;
    \node[draw,circle] (b) at (0,-0.75) {} ;
    \node[draw,circle] (c) at (-0.5,-1.75) {} ;
    \node[draw,circle] (d) at (0.5,-1.75) {} ;
    \node[draw,circle] (e) at (-1,-2) {} ;
    \node[draw,circle] (f) at (1,-2) {} ;

    \draw (a) -- (b) -- (c) -- (d) -- (b) ;
    \draw (c) -- (e) ;
    \draw (d) -- (f) ;
  \end{tikzpicture}
  \caption{The net.}
  \label{fig:net}
  \end{subfigure}
   \begin{subfigure}[b]{0.3\textwidth}
  \begin{tikzpicture}
    \node[draw,circle] (a) at (0,0) {} ;
    \node[draw,circle] (b) at (1,0) {} ;
    \node[draw,circle] (c) at (2,0) {} ;
    \node[draw,circle] (d) at (3,0) {} ;
    \node[draw,circle] (e) at (4,0) {} ;
    
    \node[draw,circle] (f) at (0.5,-1) {} ;
    \node[draw,circle] (g) at (1.5,-1) {} ;
    \node[draw,circle] (h) at (2.5,-1) {} ;
    \node[draw,circle] (i) at (3.5,-1) {} ;

    \draw (e) -- (d) -- (c) -- (b) -- (a) -- (f) -- (b) -- (g) -- (c) -- (h) -- (d) -- (i) ;
    \draw[dashdotted] (i) -- (e) ;
  \end{tikzpicture}
  \caption{The chain of four triangles with or without the dash-dotted edge.}
  \label{fig:triangles}
   \end{subfigure}
   \qquad
  \begin{subfigure}[b]{0.3\textwidth}
  \begin{tikzpicture}
    \node[draw,circle] (a) at (0,-2) {} ;
    \node[draw,circle] (b) at (1,0) {} ;
    \node[draw,circle] (c) at (2,-2) {} ;
    \node[draw,circle] (d) at (0.5,-1) {} ;
    \node[draw,circle] (e) at (1.5,-1) {} ;
    \node[draw,circle] (f) at (1,-2) {} ;

    \draw (a) -- (f) -- (c) -- (e) -- (d) -- (b) ;
    \draw (e) -- (b) ;
    \draw (e) -- (f) -- (d) -- (a) ;
  \end{tikzpicture}
  \caption{The triforce.}
  \label{fig:unnamed}
  \end{subfigure}
  \qquad
   \begin{subfigure}[b]{0.25\textwidth}
  \begin{tikzpicture}
    \node[draw,circle] (a) at (0,0) {} ;
    \node[draw,circle] (b) at (1,0) {} ;
    \node[draw,circle] (c) at (2,0) {} ;
    \node[draw,circle] (d) at (0.5,-1) {} ;
    \node[draw,circle] (e) at (1.5,-1) {} ;
    \node[draw,circle] (f) at (1,1) {} ;

    \draw (a) -- (b) -- (c) -- (e) -- (d) ;
    \draw (b) -- (f) ;
    \draw (e) -- (b) -- (d) -- (a) ;
  \end{tikzpicture}
  \caption{$\overline{\text{Gem} \cup K_1}$ .}
  \label{fig:net}
   \end{subfigure}
    \qquad
   \begin{subfigure}[b]{0.25\textwidth}
  \begin{tikzpicture}
    \node[draw,circle] (a) at (0,0) {} ;
    \node[draw,circle] (b) at (1,0) {} ;
    \node[draw,circle] (c) at (2,0) {} ;
    \node[draw,circle] (d) at (0.5,-1) {} ;
    \node[draw,circle] (e) at (1.5,-1) {} ;
    \node[draw,circle] (f) at (1,1) {} ;

    \draw (a) -- (b) -- (c) -- (e) ;
    \draw (b) -- (f) ;
    \draw (e) -- (b) -- (d) -- (a) ;
  \end{tikzpicture}
  \caption{$T_{1,2,2}=\overline{W_4 \cup K_1}$.}
  \label{fig:t122}
  \end{subfigure}
  \caption{Some connected chordal $K_{1,4}$-free graphs $H$ for which $H$-free \mis is $W[1]$-hard (see \cite{BonnetBCTW18}). These graphs do not fit the candidate forms of Figure~\ref{fig:FPTcandidates} for subtle reasons and illustrate how delicate the parameterized dichotomy promises to be. In particular, observe that \mis is $W[1]$-hard on $T(1,2,2)$-graphs, whereas we will show in this paper that it is FPT in $T(1,1,2)$-graphs (a.k.a. cricket-free graphs).}
  \label{fig:tough-customers}
\end{figure}

We show in this paper that \mis is (randomized) FPT in $T(1,1,2)$-free graphs (or cricket-free graphs).
This is in sharp contrast with the $W[1]$-hardness for $T(1,2,2)$-free graphs \cite{BoBoChThWa18} (see Figure~\ref{fig:t122}).
It disproves a seemingly reasonable conjecture that FPTness is preserved by adding a true twin to a vertex of $H$.
We thus have a fairly good understanding of the parameterized complexity of \mis when $H$ is obtained by substituting cliques on a claw.
We therefore turn towards the graphs $H$ obtained by substituting cliques on a path.
\mis was shown FPT on \pk{t,t,t}-free graphs \cite{BonnetBCTW18}.
A natural next step is to attack the following conjecture.

\begin{conjecture}\label{conj:ptttt}
  For any integer $t$, \mis can be solved in FPT time in \pk{t,t,t,t}-free graphs.
\end{conjecture}

We denote by $P_\ell(t)$ the graph \pk{t,t,\ldots,t} where the sequence $t,t, \ldots, t$ is of length $\ell$.
We further conjecture the following, which is a far more distant milestone.
\begin{conjecture}\label{conj:ptGen}
For any integers $t$ and $\ell$, \mis is FPT in $P_\ell(t)$-free graphs.
\end{conjecture}
Let us recall though that the parameterized complexity of \mis is open in $P_7$-free graphs, and no \emph{easy} FPT algorithm is known on $P_5$-free graphs.
In general, we believe that there will be very few connected candidates (as described by Figure~\ref{fig:FPTcandidates}) which will not end up in (randomized) FPT.
As a first empirical evidence, we show that the four candidates remaining among the 34 graphs on five vertices indeed all lead to (randomized) FPT algorithms.

\begin{figure}[h!]
  \centering
\begin{subfigure}[b]{0.24\textwidth}
 \begin{tikzpicture} [scale=1]
 \node[draw, circle] (a) at (0,0) {} ;
 \node[draw, circle] (b) at (0,1) {} ;
 \node[draw, circle] (c) at (1,0) {} ;
 \node[draw, circle] (d) at (1,1) {} ;
 \node[draw, circle] (e) at (2,0.5) {} ;
 \draw (a) -- (b);
 \draw (d) -- (a);
 \draw (a) -- (c);
 \draw (d) -- (b);
 \draw (e) -- (c);
  
   \end{tikzpicture} 
   \caption{$\bar{P}$}
\end{subfigure}
\begin{subfigure}[b]{0.24\textwidth}
 \begin{tikzpicture} [scale=1]
 \node[draw, circle] (a) at (0,0) {} ;
 \node[draw, circle] (b) at (0,1) {} ;
 \node[draw, circle] (c) at (1,0) {} ;
 \node[draw, circle] (d) at (1,1) {} ;
 \node[draw, circle] (e) at (2,0.5) {} ;
 \draw (d) -- (b);
 \draw (d) -- (a);
 \draw (a) -- (b);
 \draw (c) -- (a);
 \draw (e) -- (c);
 \draw (d) -- (c);
  
    \end{tikzpicture}
    \caption{Kite}
\end{subfigure}
\begin{subfigure}[b]{0.24\textwidth}
 \begin{tikzpicture} [scale=1]
 \node[draw, circle] (a) at (0,0) {} ;
 \node[draw, circle] (b) at (0,1) {} ;
 \node[draw, circle] (c) at (1,0) {} ;
 \node[draw, circle] (d) at (1,1) {} ;
 \node[draw, circle] (e) at (2,0.5) {} ;
 \draw (d) -- (b);
 \draw (d) -- (a);
 \draw (a) -- (b);
 \draw (c) -- (a);
 \draw (e) -- (d);
 \draw (d) -- (c);
  
   \end{tikzpicture} 
   \caption{Dart}
\end{subfigure}
\begin{subfigure}[b]{0.24\textwidth}
 \begin{tikzpicture} [scale=1]
 \node[draw, circle] (a) at (0,0) {} ;
 \node[draw, circle] (b) at (0,1) {} ;
 \node[draw, circle] (c) at (1,0) {} ;
 \node[draw, circle] (d) at (1,1) {} ;
 \node[draw, circle] (e) at (2,0.5) {} ;
 \draw (a) -- (b);
 \draw (d) -- (a);
 \draw (d) -- (c);
 \draw (d) -- (b);
 \draw (e) -- (d);
  
   \end{tikzpicture}
 \caption{Cricket}
 \label{fig:cricket}
\end{subfigure}
\caption{The four (out of $34$) remaining cases on five vertices for the FPT/$W[1]$-hard dichotomy (see \cite{BoBoChThWa18}). In this paper, we come up with new tools and solve all of them in (randomized) FPT.}
\label{fig:remainingfour}
\end{figure}
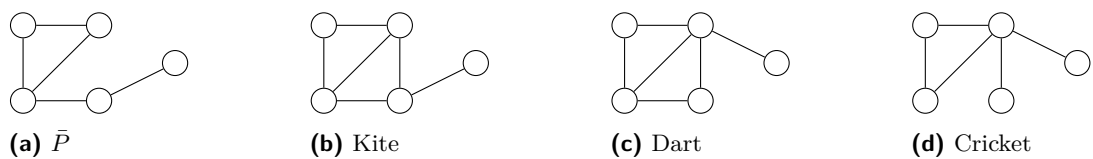

\subparagraph*{Organization of our results.}
The rest of the paper is organized as follows.
In Section~\ref{sec:prelim}, we introduce FPT Turing-reductions relevant to the subsequent section.
In Section~\ref{sec:almostUnionOrSum}, we give a series of FPT algorithms in far-reaching generalizations of co-graphs: graphs where the deletion of $g(k)$ vertices leads to a separation which is either very sparse or very dense, in a way that is controlled by the parameter. 
In Section~\ref{sec:p1ttt}, we use these results to obtain an FPT algorithm on \pk{1,t,t,t}-free graphs for any positive integer $t$, taking a stab at \cref{conj:ptttt}.
Observe that this result settles at the same time $\overline P$ ($=$\pk{1,1,1,2}) and the kite ($=$\pk{1,1,2,1}). 
In Section~\ref{sec:dartAndCricket}, we finish the FPT/$W[1]$-hard classification for five-vertex graphs by designing randomized FPT algorithms on dart-free graphs and cricket-free graphs.

We believe that the results of Section~\ref{sec:almostUnionOrSum} as well as the techniques developed in~\cref{sec:p1ttt,sec:dartAndCricket} may help in settling \cref{conj:ptttt}.
For \pk{t,t,t,t,t}-free graphs, it is possible that one will have to combine the framework of potential maximal cliques with our techniques.
To solve \cref{conj:ptGen}, let alone the full parameterized dichotomy, some new ideas will be needed. 
The FPT algorithms of the current paper merely serve for classification purposes, and are not practical.
A possible line of work is to get improved running times for the already established FPT cases.
We also hope that the results of Section~\ref{sec:almostUnionOrSum} will prove useful in a context other than $H$-free graphs.

\section{Preliminaries}\label{sec:prelim}

Here, we introduce some basics about graph notations, Ramsey numbers, and FPT algorithms.

\subsection{Notations}
For any pair of integers $i \leqslant j$, we denote by $[i,j]$ the set of integers $\{i,i+1,\ldots,j-1,j\}$, and for any positive integer $i$, $[i]$ is a shorthand for $[1,i]$.
We use the standard graph terminology and notations \cite{Die12}.
All our graphs are finite and simple, \ie, they have no multiple edge nor self-loop.
For a vertex $v$, we denote by $N(v)$ the set of neighbors of $v$, and $N[v] := N(v) \cup \{v\}$.
For a subset of vertices $S$, we set $N(S) := \bigcup_{v \in S}N(v) \setminus S$ and $N[S] := N(S) \cup S$. The \emph{degree} (resp. \emph{co-degree}) of a vertex $v$ is $|N(v)|$ (resp. $|V \setminus N[v]|$).
If $G$ is a graph and $X$ is a subset of its vertices, $G[X]$ is the subgraph induced by $X$ and $G-X$ is a shorthand for $G[V(G) \setminus X]$.
We denote by $\alpha(G)$ the independence number, that is the size of a maximum independent set.
If $H$ and $G$ are two graphs, we write $H \subseteq_i G$ to mean that $H$ is an induced subgraph of $G$, and $H \subset_i G$ if $H$ is a proper induced subgraph of $G$.
We denote by $K_\ell$, $P_\ell$, $C_\ell$, the clique, path, cycle, respectively, on $\ell$ vertices, and by $K_{s,t}$ the complete bipartite graph with $s$ vertices on one side and $t$, on the other.
The claw is $K_{1,3}$, and the paw is the graph obtained by adding one edge to the claw. 
If $H$ is a graph and $t$ is a positive integer, we denote by $tH$ the graph made of $t$ disjoint copies of $H$.
For instance, $2K_2$ corresponds to the disjoint union of two edges.
We say that a class of graphs $\mathcal C$ is \emph{hereditary} if it is closed by induced subgraph, \ie, $\forall H, G$, $(G \in \mathcal C \land H \subseteq_i G) \Rightarrow H \in \mathcal C$.

\subsection{Ramsey numbers}
For two positive integers $a$ and $b$, $R(a,b)$ is the smallest integer such that any graph with at least that many vertices has an independent set of size $a$ or a clique of size $b$.
By Ramsey's theorem, $R(a,b)$ always exists and is no greater than ${a+b \choose a}$.
For the sake of convenience, we set $\Ramsey(a,b) := {a+b \choose a} = {a+b \choose b}$.
We will use repeatedly a constructive version of Ramsey's theorem.
\begin{lemma}[folklore]\label{lem:constructiveRamsey}
Let $a$ and $b$ be two positive integers, and let $G$ be a graph on at least $\Ramsey(a, b)$ vertices. Then an independent set of size $a$ or a clique of size $b$ can be found in linear time.
\end{lemma}
\begin{proof}
  We show this lemma by induction on $a+b$.
  For $a=1$ (or $b=1$), any vertex of $G$ works (it is a clique and an independent set at the same time).
  And $G$ is non-empty since it has at least ${1+b \choose 1}$ (or ${a+1 \choose 1}$) vertices.
  We assume $a, b \geqslant 2$ and consider any vertex $v$ of $G$.
  Let $G_1 := G - N[v]$ and $G_2 := G[N(v)]$, so $|V(G)|=1+|V(G_1)|+|V(G_2)|$.

  Since $|V(G)| \geqslant {a+b \choose a} = {a+b-1 \choose a-1} + {a+b-1 \choose a}$, it cannot be that both $|V(G_1)| \leqslant \Ramsey(a-1,b) - 1$ and $|V(G_2)| \leqslant \Ramsey(a,b-1) - 1$.
  If $G_1$ has at least $\Ramsey(a-1,b)$ vertices, we find by induction an independent set $I$ of size $a-1$ or a clique of size $b$.
  Thus $I \cup \{v\}$ is an independent set of size $a$ in $G$.
  If instead $G_2$ has at least $\Ramsey(a,b-1)$ vertices, we find by induction an independent set of size $a$ or a clique $C$ of size $b-1$.
  Thus $C \cup \{v\}$ is an independent set of size $b$ in $G$.
\end{proof}
For two positive integers $a$ and $b$, we denote by $\Ramsey_a(b)$ the smallest integer $n$ such that any edge-coloring of $K_n$ with $a$ colors has a monochromatic clique of size $b$.
In particular, $\Ramsey_2(b)=\Ramsey(b,b)$ (one color for the edges, and one color for the non-edges).
Again, $\Ramsey_a(b)$ always exists and a monochromatic clique of size $b$ in an $a$-edge-colored clique of size at least $\Ramsey_a(b)$ can be found in polynomial-time (whose exponent does \emph{not} depend on $a$ and $b$). 

\subsection{FPT Turing-reductions}
For an instance $(I, k)$ of \mis, let $\text{yes}(I,k)$ be the Boolean function which equals $True$ if and only if $(I, k)$ is a positive instance.
\begin{definition}\label{def:TuringReduction}
 A \emph{decreasing FPT $g$-Turing-reduction} is an FPT algorithm which, given an instance $(I,k)$, produces $\ell := g(k)$ instances $(I_1,k_1), \ldots, (I_\ell,k_\ell)$, for some computable function $g$, such that:
  \begin{itemize}
  \item(i) $\text{yes}(I,k) \Leftrightarrow \phi(\text{yes}(I_1,k_1), \ldots, \text{yes}(I_\ell,k_\ell))$, where $\phi$ is a fixed FPT-time checkable formula\footnote{By \textit{FPT-time checkable formula}, we mean that there exists an algorithm which takes as input $\ell$ Booleans $b_1$, $\dots$, $b_{\ell}$ and tests whether $\phi(b_1, \dots, b_\ell)$ is true in FPT time parameterized by $\ell$.}, 
  \item(ii) $|I_j| \leqslant |I|$ for every $j \in [\ell]$, and
  \item(iii) $k_j \leqslant k-1$ for every $j \in [\ell]$. 
  \end{itemize}
\end{definition}
Note that conditions (ii) and (iii) prevent the instance size from increasing and force the parameter to strictly decrease, respectively.

\begin{lemma}\label{lem:TuringHered}
  Assume there is a decreasing FPT $g$-Turing-reduction for \mis on every input $(G \in \mathcal C, k)$, running in time $h(k)|V(G)|^\gamma$ (this includes the time to check $\phi$).
  Let $f : [k-1] \rightarrow \mathbb{N}$ be a non-decreasing function.
  If any instance $(H,k')$ with $k' < k$ can be solved in time $f(k')|V(H)|^c$ with $c \geqslant \gamma$,
  then \mis can be solved in FPT time $f(k)|V(G)|^c$ in $\mathcal C$, with $f(k) := h(k)+g(k)f(k-1)$.
\end{lemma}
\begin{proof}
  We show the lemma by induction.
  If $k=1$, this is immediate.
  We therefore assume that $k \geqslant 2$.
  We apply the decreasing FPT $g$-Turing-reduction to $(G,k)$.
  That creates at most $g(k)$ instances with parameter at most $k-1$.
  We solve each instance in time $f(k-1)n^c$ with $n := |V(G)|$.
 The overall running time is bounded by $h(k)n^\gamma+g(k)f(k-1)n^c \leqslant f(k)n^c$ by extending the partial function $f$ with $f(k) := h(k)+g(k)f(k-1)$.
\end{proof}
This corollary follows by induction on the parameter $k$.
\begin{corollary}\label{cor:TuringHered}
  If \mis admits a decreasing FPT $g$-Turing-reduction on a hereditary class, then \mis can be solved in FPT time in $\mathcal C$.
\end{corollary}

\begin{definition}\label{def:TuringReduction2}
  An \emph{improving FPT $g$-Turing-reduction} is an FPT time $h(k)|V(G)|^\gamma$ algorithm which, given an instance $(I,k)$, produces some instances $(I_1,k_1), \ldots, (I_\ell,k_\ell)$, and can check a formula $\phi$, such that:
  \begin{itemize}
  \item(i) $\text{yes}(I,k) \Leftrightarrow \phi(\text{yes}(I_1,k_1), \ldots, \text{yes}(I_\ell,k_\ell))$, and
  \item(ii) $\exists c_0, f_0$, $\forall c \geqslant c_0, f \in \Omega(f_0)$, $h(k)|V(G)|^\gamma + \sum\limits_{j \in [\ell]} f(k_j)|I_j|^c \leqslant f(k)|I|^c$.
  \end{itemize}
\end{definition}
\begin{lemma}\label{lem:TuringHered2}
  Assume there is an improving FPT $g$-Turing-reduction for \mis on every input $(I \in \mathcal C, k)$, producing in time $h(k)|I|^\gamma$, some instances $(I_1,k_1), \ldots, (I_\ell,k_\ell)$. 
  If each instance $(I_j,k_j)$ can be solved in time $h(k_j)|I_j|^{c'}$, then \mis can be solved in FPT time in $\mathcal C$.
\end{lemma}
\begin{proof}
  Let $c := \max(c_0,c')$ and $f := \max(f_0,h)$, for the $c_0$ and $f_0$ of Definition~\ref{lem:TuringHered2}.
  \emph{A fortiori}, instances $(I_j,k_j)$ can be solved in time $f(k_j)|I_j|^c$.
  We call the Turing-reduction on $(I,k)$, solve every subinstances $(I_j,k_j)$, and check $\phi$.
  By item (ii), the overall running time $h(k)|V(G)|^\gamma + \sum\limits_{j \in [\ell]} f(k_j)|I_j|^c$ is bounded by $f(k)|I|^c$.
  By item (i), this decides $(I,k)$.
\end{proof}


When trying to compute \mis in FPT time, one can assume that there is no vertex of bounded degree or bounded co-degree (in terms of a function of $k$).

\begin{observation}\label{obs:boundedDegree}
  Let $(G,k)$ be an input of \mis with a vertex $v$ of degree $g(k)$ for some computable function $g$.
  Then the instance admits a decreasing FPT Turing-reduction.
\end{observation}
\begin{proof}
  A maximal independent set has to intersect $N[v]$.
  So, we can branch on $g(k)+1$ instances with parameter $k-1$.
\end{proof}

\begin{observation}\label{obs:boundedCoDegree}
Let $(G,k)$ be an input of \mis with a vertex $v$ of co-degree $g(k)$ for some computable function $g$.
Then the instance admits an improving FPT Turing-reduction.
\end{observation}
\begin{proof}
  We can find the vertex $v$ in time $ng(k)$ with $n := |V(G)|$, and we assume $n \geqslant 2$.
  By branching on $v$, we define two instances $(G - N[v],k-1)$ and $(G-\{v\},k)$ (which corresponds to including $v$ to the solution, or not).
  The first instance can be further reduced in time $g(k)^{k-1}$ (by actually solving it).
  So the two instances output by the Turing-reduction are $\text{Bool}$ and $(G-\{v\},k)$, where $Bool$ is the result of solving $(G - N[v],k-1)$.
  The formula $\phi$ is just $\text{Bool} \lor \text{yes}(G-\{v\},k)$.
  Let $c_0 := 2$ and $f_0(k) := g(k)^{k-1}$.
  For all $c \geqslant c_0$ and $f \in \Omega(f_0)$, $ng(k)+g(k)^{k-1}+f(k)(n-1)^c \leqslant nf(k)+f(k)+f(k)(n-1)^c \leqslant f(k)(n+1+(n-1)^c) \leqslant f(k)n^c$.
\end{proof}

\section{Almost disconnected and almost join graphs}\label{sec:almostUnionOrSum}

We say that a graph is a join or a \emph{complete sum}, if there is a non-trivial bipartition $(A,B)$ of its vertex set (\ie $A$ and $B$ are non-empty) such that every pair of vertices $(u,v) \in A \times B$ is linked by an edge.  
Equivalently, a graph is a complete sum if its complement is disconnected.
In the following subsection, we define a series of variants of complete sums and disjoint unions in the presence of a \emph{parameterized noise}.

\subsection{Definition of the classes}

In all the following definitions, we say that a tripartition $(A, B, R)$ is \emph{non-trivial} if $A$ and $B$ are non-empty and $|R| < \min(|A|,|B|)$.
Notice that we do not assume $R$ is non-empty.

\begin{definition}\label{def:almostDisconnected}
Graphs in a class $\mathcal C$ are \emph{$(g,d)$-almost disconnected} if there exist two computable functions $g$ and $d$, such that for every $G \in \mathcal C$ and $k \geqslant \alpha(G)$, there is a non-trivial tripartition $(A,B,R)$ of $V(G)$ satisfying:
  \begin{itemize}
  \item $|R| \leqslant g(k)$, and
  \item $\forall v \in A$, $|N(v) \cap B| \leqslant d(k)$ and $\forall v \in B$, $|N(v) \cap A| \leqslant d(k)$.
  \end{itemize}
\end{definition}

Similarly, we define a generalization of a complete sum.
\begin{definition}\label{def:almostComplete}
Graphs in a class $\mathcal C$ are \emph{$(g,d)$-almost bicomplete} if there exist two computable functions $g$ and $d$, such that for every $G \in \mathcal C$ and $k \geqslant \alpha(G)$, there is a non-trivial tripartition $(A,B,R)$ of $V(G)$ satisfying:
  \begin{itemize}
  \item $|R| \leqslant g(k)$, and
  \item $\forall v \in A$, $|B \setminus N(v)| \leqslant d(k)$ and $\forall v \in B$, $|A \setminus N(v)| \leqslant d(k)$.
  \end{itemize}
\end{definition}

By extension, if $\mathcal{C}$ only contains graphs which are almost disconnected (resp. $(g, d)$-almost disconnected, almost bicomplete, $(g, d)$-almost bicomplete), then we say that $\mathcal{C}$ is almost disconnected (resp. $(g, d)$-almost disconnected, almost bicomplete, $(g, d)$-almost bicomplete).
Note that we do not require an almost disconnected or an almost bicomplete class to be hereditary.
For $G \in \mathcal C$, we call a satisfying tripartition $(A,B,R)$ a \emph{witness of almost disconnectedness} (resp. \emph{witness of almost bicompleteness}).

We define the one-sided variants.
\begin{definition}\label{def:almostDisconnected1sided}
Graphs in a class $\mathcal C$ are \emph{one-sided $(g,d)$-almost disconnected} if there exist two computable functions $g$ and $d$, such that for every $G \in \mathcal C$ and $k \geqslant \alpha(G)$, there is a non-trivial tripartition $(A,B,R)$ of $V(G)$ satisfying:
  \begin{itemize}
  \item $|R| \leqslant g(k)$,
  \item $|B| > kd(k)$, and
  \item $\forall v \in A$, $|N(v) \cap B| \leqslant d(k)$.
  \end{itemize}
\end{definition}

In the above definition, the second condition is purely a technical one.
Observe, though, that any tripartition $(A,B,R)$ with $|R| < |B| \leqslant d(k)$ trivially satisfies the third condition (provided $|R| < d(k)$).
So a condition forcing $B$ to have more than $d(k)$ vertices is somehow needed.
Now, we set the lower bound on $|B|$ a bit higher to make Lemma~\ref{lem:almostDisjoint} work.
Similarly, we could define the one-sided generalization of a complete sum.

\begin{definition}\label{def:almostComplete1sided}
Graphs in a class $\mathcal C$ are \emph{one-sided $(g,d)$-almost bicomplete} if there exist two computable functions $g$ and $d$, such that for every $G \in \mathcal C$ and $k \geqslant \alpha(G)$, there is a non-trivial tripartition $(A,B,R)$ of $V(G)$ satisfying:
  \begin{itemize}
  \item $|R| \leqslant g(k)$,
  \item if there is an independent set of size $k$, there is one that intersects $A$, and
  \item $\forall v \in B$, $|A \setminus N(v)| \leqslant d(k)$.
  \end{itemize}
\end{definition}
Again, the second condition is there to make Theorem~\ref{thm:withoutBiclique} work.

\subsection{Improving and decreasing FPT Turing-reductions}

The following technical lemma will be used to bound the running time of recursive calls on two almost disjoint parts of the input.
\begin{lemma}\label{lem:almostSplit}
  Suppose $\gamma \geqslant 0$ and $c \geqslant \max(2,\gamma+2)$ are two constants, and $n_1,n_2,n,u$ are four positive integers such that $n_1+n_2+u = n$ and $\min(n_1,n_2)>u$.
  Then, $$n^\gamma+(n_1+u)^c+(n_2+u)^c < n^c.$$
\end{lemma}
\begin{proof}
  First we observe that $n^2-((n_1+u)^2+(n_2+u)^2)=n_1^2+n_2^2+u^2+2(n_1n_2+n_1u+n_2u)-(n_1^2+2n_1u+u^2+n_2^2+2n_2u+u^2)=2n_1n_2 - u^2 > 2u^2-u^2 = u^2 \geqslant 1$.
  Now, $n^c=n^{c-2}n^2 \geqslant n^{c-2}(1+(n_1+u)^2+(n_2+u)^2) \geqslant n^{c-2}(n^{\gamma-c+2}+(n_1+u)^2+(n_2+u)^2) = n^\gamma + n^{c-2}(n_1+u)^2 + n^{c-2}(n_2+u)^2 > n^\gamma + (n_1+u)^c + (n_2+u)^c$.
  The last inequality holds since $n > n_1+u$ and $n > n_2+u$.
\end{proof}

We start with an improving FPT Turing-reduction on almost bicomplete graphs.
It finds a kernel for solutions intersecting both $A$ and $B$, solves recursively on $A \cup R$ and $B \cup R$ for the other solutions, and uses Lemma~\ref{lem:almostSplit} to bound the overall running time.

\begin{lemma}\label{lem:completeSum}
  Let $\mathcal C$ be a $(g,d)$-almost bicomplete class of graphs.
  Suppose for every $G \in \mathcal C$, a witness $(A,B,R)$ of almost bicompleteness can be found in time $h(k)|V(G)|^\gamma$.
  Then, \mis admits an improving FPT Turing-reduction in $\mathcal C$.
  In particular, \mis can be solved in FPT time if both $(G[A \cup R],k)$ and $(G[B \cup R],k)$ can.
\end{lemma}

\begin{proof}
  We can detect a potential solution $S$ intersecting both $A$ and $B$ in time $n^2(2d(k)+g(k))^k = n^2s(k)$, with $n := |V(G)|$, by setting $s(k) := (2d(k)+g(k))^{k-2}$.
  We exhaustively guess one vertex $a \in S \cap A$ and one vertex $b \in S \cap B$.
  For each of these quadratically many choices, there are at most $d(k)$ non-neighbors of $a$ in $B$ and at most $d(k)$ non-neighbors of $b$ in $A$.
  So the remaining instance $G - (N(a) \cup N(b))$ has at most $2d(k)+g(k)$ vertices; hence the running time.

  We are now left with potential solutions intersecting $A$ but not $B$, or $B$ but not $A$.
  These are fully contained in $A \cup R$ or in $B \cup R$.
  Let $n_1 := |A|$ and $n_2 := |B|$ (so $n = n_1+n_2+|R|$).
  The two last branches just consist of recursively solving the instances $(G[A \cup R],k)$ and $(G[B \cup R],k)$.
  Let $c_0 := \max(4,\gamma+2)$ and $f_0 := h+s$.
  For all $c \geqslant c_0$ and $f \in \Omega(f_0)$, 
  $$h(k)n^\gamma+s(k)n^2+f(k)(n_1+g(k))^c+f(k)(n_2+g(k))^c$$
  $$\leqslant f(k)n^{\max(\gamma,2)} + f(k)(n_1+g(k))^c+f(k)(n_2+g(k))^c \leqslant f(k)n^c.$$
  The last inequality holds by Lemma~\ref{lem:almostSplit}, since $\max(\gamma,2)+2 \leqslant c$ and $\min(n_1,n_2) > g(k)$.
  The conclusion holds by Lemma~\ref{lem:TuringHered2}.
\end{proof}

If we only have \emph{one-sided} almost bicompleteness, we need some additional conditions on the solution: at least one solution should intersect $A$ (see Definition~\ref{def:almostComplete1sided}).
We recall that $H \subset_i G$ means that $H$ is a proper induced subgraph of $G$.
\begin{lemma}\label{lem:completeSum1sided}
  Let $\mathcal C$ be a one-sided $(g,d)$-almost bicomplete class of graphs.
  Suppose for every $G \in \mathcal C$, a witness $(A,B,R)$ of one-sided almost bicompleteness can be found in time $h(k)|V(G)|^\gamma$.
  Then, \mis admits an improving FPT Turing-reduction in~$\mathcal C$.
  In particular, \mis can be solved in FPT time if $(G[A \cup R],k)$ and $\forall k' \leqslant k-1$, $\forall H \subset_i G$, $(H,k')$ all can.
\end{lemma}

\begin{proof}
  Let $S$ be an unknown solution.
  Let $k_1 := S \cap A$ and $k_2 := S \cap B$.
  Let us anticipate on an FPT running time $f(k)n^c$ for instances of size $n$ and parameter $k$ (the definition of $f$ will be given later).
  For instance, covering the case $k_2=0$ takes time $f(k)|A \cup R|^c$, since it consists in solving $(G[A \cup R],k)$.
  By assumption, we do not have to consider the case $k_1=0$.
  For each pair $k_1,k_2$ such that $k_1 \geqslant 1, k_2 \geqslant 1, k_1+k_2 \leqslant k$, we do the following.
  
  An independent set of size $k_1$ in $G[A]$ is \emph{candidate} if it is in the non-neighborhood of at least one vertex $v \in B$.
  Since $k_2 \geqslant 1$, we can restrict the search in $A$ to candidate independent sets of size $k_1$.
  Indeed, any independent set in $A$, not in the non-neighborhood of any vertex of $B$, cannot be extended to $k_2$ ($\geqslant 1$) more vertices of $B$.
  For each candidate independent set $I_1$ of size $k_1$, we compute an independent set of size $k_2$ in $B \setminus N(I_1)$.
  This takes time \[ \sum\limits_{\substack{I_1~\text{candidate} \\ |I_1|=k_1}} f(k_2)|B \setminus N(I_1)|^c=f(k_2) \sum\limits_{\substack{I_1~\text{candidate} \\ |I_1|=k_1}} |B \setminus N(I_1)|^c \leqslant f(k_2) \left( \sum\limits_{\substack{I_1~\text{candidate} \\ |I_1|=k_1}} |B \setminus N(I_1)| \right) ^c \]
  by Cauchy-Schwarz inequality (since $c \geqslant 2$).
  Now, since $k_1 > 0$,
  \[
  \sum\limits_{\substack{I_1~\text{candidate} \\ |I_1|=k_1}} |B \setminus N(I_1)| \leqslant \sum\limits_{\substack{I_1~\text{candidate} \\ |I_1|=k_1}} |I_1| \cdot |B \setminus N(I_1)| \leqslant {d(k) \choose k_1}d(k)|B|.
  \]
  The last inequality holds since $\sum_{I_1~\text{candidate}, |I_1|=k_1} |I_1| \cdot |B \setminus N(I_1)|$ counts the number of non-edges between $A$ and $B$ with multiplicity at most ${d(k) \choose k_1}$.
  Indeed a same non-edge $uv$ (with $u \in A$, $v \in B$) is counted for at most ${d(k) \choose k_1}$ candidate independent sets (since they have to be in the non-neighborhood of $v$).
  Since, by assumption, vertices in $B$ have at most $d(k)$ non-neighbors in $A$, the total number of non-edges is $d(k)|B|$.
  Let $c_0 \geqslant \gamma + 2$ and $f_0 := \max(h,k \mapsto k^{2k}{d(k) \choose k}^{ck}d(k)^{ck})$.
  For any $c \geqslant c_0$ and $f \in \Omega(f_0)$,
  $$h(k)|V(G)|^\gamma + f(k)|A \cup R|^c + \sum\limits_{\substack{k_1 \in [k-1], k_2 \in [k-k_1]}} f(k_2) \left( \sum\limits_{\substack{I_1~\text{candidate} \\ |I_1|=k_1}} |B \setminus N(I_1)| \right) ^c$$ $$\leqslant h(k)|V(G)|^\gamma + f(k)|A \cup R|^c + k^2f(k-1) {d(k) \choose k}^cd(k)^c|B|^c$$ $$\leqslant f(k)|V(G)|^\gamma + f(k)|A \cup R|^c + f(k)|B|^c \leqslant f(k)|V(G)|^c$$
    since $f(k) \geqslant k^2{d(k) \choose k}^cd(k)^c f(k-1)$.
    The last inequality holds by Lemma~\ref{lem:almostSplit}.
    The conclusion holds by Lemma~\ref{lem:TuringHered2}.
\end{proof}

We now turn our attention to almost disconnected classes.
For these classes, we obtain decreasing FPT Turing-reductions, \ie, where the produced instances have a strictly smaller parameter than the original instance.

\begin{lemma}\label{lem:almostDisjoint}
  Let $\mathcal C$ be a one-sided $(g,d)$-almost disconnected class of graphs.
  Suppose for every $G \in \mathcal C$, a witness $(A,B,R)$ of one-sided almost disconnectedness can be found in time $h(k)|V(G)|^\gamma$.
  Then, \mis admits a decreasing FPT Turing-reduction in $\mathcal C$.
  In particular, \mis can be solved in FPT time if $\forall k' \leqslant k-1$ and $\forall H \subseteq_i G$, instances $(H,k')$ can. 
\end{lemma}

\begin{proof}
  Let $S$ be an unknown but supposed independent set of $G$ of size $k$.
  In time $h(k)n^c$ with $n := |V(G)|$, we compute a witness $(A,B,R)$.
  For each $u \in R$, we branch on including $u$ to our solution.
  This represents at most $g(k)$ branches with parameter $k-1$.
  Now, we can focus on the case $S \cap R = \emptyset$.

  We first deal separately with the special cases of $|S \cap A| = k$, $|S \cap B|= 0$ (a), and of $|S \cap A| = 0$, $|S \cap B|= k$ (b).  
  As by assumption $|B| > kd(k)$, no maximal independent set has $k$ vertices in $A$ and zero in $B$.
  Indeed, by the one-sided almost disconnectedness, any $k$ vertices in $A$ dominate at most $k^2$ vertices in $B$.
  Hence at least one vertex of $B$ could be added to this independent set of size $k$.
  So case (a) is actually impossible.

  For case (b), we proceed as follows.
  We compute an independent set of size $k-1$ in $G[B]$.
  We temporary remove it from the graph, without removing its neighborhood.
  We compute a second independent set of size $k-1$ in $G[B]$ (without the first independent set); then a third one (in the graph deprived of the first two).
  We iterate this process until no independent set of size $k-1$ is found or we reach a total of $d(k)+1$ (disjoint) independent sets of size $k-1$ excavated in $B$.
  If we stop because of the former alternative, we know that an independent set of size $k$ (actually even of size $k-1$) in $B$ has to intersect the union of at most $d(k)$ independent sets of size $k-1$; so at most $(k-1)d(k)$ vertices in total.
  In that case, we branch on each vertex of this set of size at most $(k-1)d(k)$ with parameter $k-1$.
  If we stop because of the latter condition, we can include an arbitrary vertex $w$ of $A$ in the solution.
  By assumption, $w$ has at least one neighbor in at most $d(k)$ independent sets of size $k-1$ in $B$.
  So at least one independent set of size $k-1$ of the collection is not adjacent to $w$, and forms with $w$ a solution.

  Now we are done with cases (a) and (b), we can assume that $k_1 := |S \cap A|$, $k_2 := |S \cap B| = k-k_1$ are both non-zero.
  Equivalently, $1 \leqslant k_1 \leqslant k-1$. 
  We try out all the $k-1$ possibilities.
  For each, we perform a similar trick to the one used for case (b).
  We compute an independent set $I_1$ of size $k_2$ in $G[B]$.
  We then compute an independent set $I_2$ of size $k_2$ in $G[B \setminus I_1]$.
  Observe that there may be edges between $I_1$ and $I_2$.
  We compute an independent set $I_3$ in $G[B \setminus (I_1 \cup I_2)]$, and so on.
  We iterate this process until no independent set of size $k_2$ is found or we reach a total of $d(k)k_1+1$ (disjoint) independent sets of size $k_2$ excavated in $B$.

  Say, we end up with the sets $I_1, \ldots, I_s$.
  Let $I := \bigcup_{j \in [s]} I_j$.
  If $s \leqslant f(k)k_1$, then we stopped because there was no independent set of size $k_2$ in $G[B \setminus I]$.
  This means that $S$ intersects $I$.
  In that case, we branch on each vertex of $I$.

  The other case is that $s=f(k)k_1+1$ and we stopped because we had enough sets $I_j$.
  In that case, we compute one independent set $A_1$ of size $k_1$ in $G[A]$.
  By assumption, $|N_B(A_1)| \leqslant k_1d(k)$.
  In particular, there is at least one $I_j$ which does not intersect $N_B(A_1)$.
  And $A_1 \cup I_j$ is our independent of size $k$.

  Our algorithm makes at most $$g(k)+d(k)+1+\sum\limits_{k_1 \in [k-1]}(d(k)k_1+1)+1 \leqslant g(k)+d(k)+2+k^2d(k)+k$$ recursive calls to instances with parameter $k-1$, and we conclude by Lemma~\ref{lem:TuringHered}.
\end{proof}

Let $\mathcal B(A,B)$ be the bipartite graph between two disjoint vertex-subsets $A$ and $B$ (ignoring the edges internal to $A$ and to $B$).
We can further generalize the previous result to tripartitions $(A,B,R)$ such that $\mathcal B(A,B)$ is $K_{d(k),d(k)}$-free.

\begin{definition}\label{def:withoutBiclique}
Graphs in a class $\mathcal C$ are $(g,d)$-\emph{\wc} if there exist two computable functions $g$ and $d$, such that for every $G \in \mathcal C$ and $k \geqslant \alpha(G)$, there is a non-trivial tripartition $(A,B,R)$ of $V(G)$ satisfying:
  \begin{itemize}
  \item $|R| \leqslant g(k)$,
  \item $|A|, |B| > \lceil d(k)^{d(k)}k^{2d(k)-1} \rceil + 1$, and
  \item $\mathcal B(A,B)$ is $K_{d(k),d(k)}$-free.
  \end{itemize}
\end{definition}

Again, if we do not require $|A|$ and $|B|$ to be larger than $d(k)$, such a tripartition may trivially exist.
We force $A$ and $B$ to be even larger than that to make the next theorem work.
We show this theorem by combining ideas of the proof of Lemma~\ref{lem:almostDisjoint} with the extremal theory result, known as K\H{o}v\'ari-S\'os-Tur\'an's theorem, that $K_{t,t}$-free $n$-vertex graphs have at most $t n^{2-\frac{1}{t}}$ edges \cite{kovari1954problem}.

\begin{theorem}\label{thm:withoutBiclique}
  Let $\mathcal C$ be a $(g,d)$-\wc class of graphs.
  Suppose for every $G \in \mathcal C$, a witness $(A,B,R)$ of \wcness can be found in time $h(k)|V(G)|^\gamma$. 
  Then, \mis admits a decreasing FPT Turing-reduction in $\mathcal C$.
  In particular, \mis can be solved in FPT time if $\forall k' \leqslant k-1$ and $\forall H \subseteq_i G$, the instance $(H,k')$ can.
\end{theorem}

\begin{proof}
  Let $S$ be an unknown solution with $k_1 := S \cap A$ and $k_2 := S \cap B = k - k_1$.
  As previously, we try out all the $k+1$ values for $k_1$, setting $k_2$ to $k-k_1$. 
  Let us first consider the $k-1$ branches in which $k_1 \neq 0$ and $k_2 \neq 0$.

  Let $s := \lceil d(k)^{d(k)}k^{2d(k)-1} \rceil + 1$.
  Using the same process as in Lemma~\ref{lem:almostDisjoint}, we compute $s$ disjoint independent sets $A_1, \ldots, A_s$ of size $k_1$ in $G[A]$ and $s$ disjoint independent sets $B_1, \ldots, B_s$ of size $k_2$ in $G[B]$.
  Again, if the process stops before we reach $s$ independent sets, we know that a solution (with $k_1$ vertices of $A$ and $k_2$ vertices of $B$) intersects a set of size at most $k_1(s-1)$ or $k_2(s-1)$ and we can branch (since $s$ is bounded by a function of $k$).
  
  Now we claim that there is at least one pair $(A_i,B_j)$ (among the $s^2$ pairs) without any edge between $A_i$ and $B_j$; hence $A_i \cup B_j$ is an independent of size $k$.
  Suppose that this is not the case.
  Then, there is at least one edge between each pair $(A_i,B_j)$.
  Therefore the bipartite graph $\mathcal B := \mathcal B(\bigcup_{i \in [s]} A_i,\bigcup_{i \in [s]} B_i)$ has at least $s^2$ edges, and $sk_1+sk_2=sk$ vertices.
  As $\mathcal B$ is also $K_{d(k),d(k)}$-free, it has, by K\H{o}v\'ari-S\'os-Tur\'an's theorem, at most $d(k)(sk)^{2-\frac{1}{d(k)}}$ edges.
  But, by the choice of $s$, $s^2 > d(k)(sk)^{2-\frac{1}{d(k)}}$, a contradiction.

  We now deal with the case $k_1=0$.
  We show that if a solution exists with $k_1=0, k_2=k$, then the branch $k_1=1, k_2=k-1$ also leads to a solution.
  Let us revisit that latter branch.
  We compute $s$ disjoint independent sets $B_1, \ldots, B_s$ of size $k-1$ in $G[B]$.
  Again, if this process stops before we reach $s$ independent sets, we can branch on each vertex of a set of size at most $(k-1)(s-1)$.
  This branching also covers the case $k_2=k$, since clearly, an independent set of size $k$ in $G[B]$ intersects those at most $(k-1)(s-1)$ vertices.
  Now, let $A'$ be any set of $s$ vertices in $A$ and $\mathcal B := \mathcal B(A', \bigcup_{i \in [s]} B_i)$.
  By applying K\H{o}v\'ari-S\'os-Tur\'an's theorem to $\mathcal B$ as in the previous paragraph, there should be at least one pair $(u,B_j) \in A' \times \{B_1, \ldots, B_s\}$ such that $u$ is not adjacent to $B_j$.
  
  We handle the case $k_2=0$ similarly, the conclusion being that we do not need to explore these branches.
  So we have described a decreasing FPT Turing-reduction creating less than $k(k+2)s$ instances (each with parameter $k'\leqslant k-1$), and we conclude by Lemma~\ref{lem:TuringHered}.
\end{proof}

A class of \emph{co-graphs with parameterized noise} is a hereditary class in which all the graphs are almost bicomplete or almost disconnected. 
The following is a direct consequence of the previous lemmas.
\begin{corollary}\label{cor:coGraphsNoise}
  Given an FPT oracle finding the corresponding tripartitions, \mis is FPT in co-graphs with parameterized noise.
\end{corollary}
The corollary still holds by replacing \emph{almost disconnected} by \emph{one-sided almost disconnected}, or even by \emph{weakly connected}. 

\subsection{Summary and usage}

Figure~\ref{fig:summaryLemmas} sums up the four FPT Turing-reductions that we obtained on almost disconnected and almost join graphs.

\begin{figure}[h!]
\begin{subfigure}[b]{0.2\textwidth}
 \begin{tikzpicture} 
   \node (Rlb) at (0,0) {} ;
   \node (Rrt) at (2.5,0.3) {} ;
   \node[draw, rounded corners, rectangle, fit=(Rlb) (Rrt)] {} ;
   \node at (1.25,0.15) {$|R| \leqslant g(k)$} ;

   \node (Alb) at (0,0.9) {} ;
   \node (Art) at (0.7,3.2) {} ;
   \node[draw, rounded corners, rectangle, fit=(Alb) (Art)] (A) {$A$} ;

   \node (Blb) at (1.8,0.9) {} ;
   \node (Brt) at (2.5,3.2) {} ;
   \node[draw, rounded corners, rectangle, fit=(Blb) (Brt)] (B) {$B$} ;

   \draw[very thick] (A) -- (B) ;

   \node[draw,circle] (u) at (0.5,1.3) {} ;
   \foreach \i in {-0.2,-0.1,0,0.1,0.2}{
     \draw[dashed] (u) --++(1.2,\i) ;
   }
   \node (nu) at (2.22,1.3) {$\leqslant d(k)$} ;

   \node[draw,circle] (v) at (2,2.8) {} ;
   \foreach \i in {-0.2,-0.1,0,0.1,0.2}{
     \draw[dashed] (v) --++(-1.2,\i) ;
   }
   \node (nv) at (0.3,2.8) {$\leqslant d(k)$} ;
 \end{tikzpicture} 
 \caption{Almost bicomplete tripartition, $A \neq \emptyset$, $B \neq \emptyset$.}
\end{subfigure}
\qquad
\begin{subfigure}[b]{0.2\textwidth}
 \begin{tikzpicture} 
   \node (Rlb) at (0,0) {} ;
   \node (Rrt) at (2.5,0.3) {} ;
   \node[draw, rounded corners, rectangle, fit=(Rlb) (Rrt)] {} ;
   \node at (1.25,0.15) {$|R| \leqslant g(k)$} ;

   \node (Alb) at (0,0.9) {} ;
   \node (Art) at (0.7,3.2) {} ;
   \node[draw, rounded corners, rectangle, fit=(Alb) (Art)] (A) {$A$} ;

   \node (Blb) at (1.8,0.9) {} ;
   \node (Brt) at (2.5,3.2) {} ;
   \node[draw, rounded corners, rectangle, fit=(Blb) (Brt)] (B) {$B$} ;

   \draw[very thick] (A) -- (B) ;

   \node[draw,circle] (v) at (2,2.8) {} ;
   \foreach \i in {-0.2,-0.1,0,0.1,0.2}{
     \draw[dashed] (v) --++(-1.2,\i) ;
   }
   \node (nv) at (0.3,2.8) {$\leqslant d(k)$} ;
 \end{tikzpicture}
 \caption{One-sided almost bicomplete, $S \cap A \neq \emptyset$ promise.}
\end{subfigure}
\qquad
\begin{subfigure}[b]{0.2\textwidth}
 \begin{tikzpicture} 
   \node (Rlb) at (0,0) {} ;
   \node (Rrt) at (2.5,0.3) {} ;
   \node[draw, rounded corners, rectangle, fit=(Rlb) (Rrt)] {} ;
   \node at (1.25,0.15) {$|R| \leqslant g(k)$} ;

   \node (Alb) at (0,0.9) {} ;
   \node (Art) at (0.7,3.2) {} ;
   \node[draw, rounded corners, rectangle, fit=(Alb) (Art)] {$A$} ;

   \node (Blb) at (1.8,0.9) {} ;
   \node (Brt) at (2.5,3.2) {} ;
   \node[draw, rounded corners, rectangle, fit=(Blb) (Brt)] {$B$} ;
   
   \node[draw,circle] (u) at (0.5,1.3) {} ;
   \foreach \i in {-0.2,-0.1,0,0.1,0.2}{
     \draw (u) --++(1.2,\i) ;
   }
   \node (nu) at (2.22,1.3) {$\leqslant d(k)$} ;
 \end{tikzpicture}
 \caption{One-sided almost disconnected, $|B| > kd(k), A \neq \emptyset$.}
\end{subfigure}
\qquad
\begin{subfigure}[b]{0.2\textwidth}
 \begin{tikzpicture} 
   \node (Rlb) at (0,0) {} ;
   \node (Rrt) at (2.5,0.3) {} ;
   \node[draw, rounded corners, rectangle, fit=(Rlb) (Rrt)] {} ;
   \node at (1.25,0.15) {$|R| \leqslant g(k)$} ;

   \node (Alb) at (0,0.9) {} ;
   \node (Art) at (0.7,3.2) {} ;
   \node[draw, rounded corners, rectangle, fit=(Alb) (Art)] {$A$} ;

   \node (Blb) at (1.8,0.9) {} ;
   \node (Brt) at (2.5,3.2) {} ;
   \node[draw, rounded corners, rectangle, fit=(Blb) (Brt)] {$B$} ;

   \foreach \j in {1,...,5}{
     \coordinate (y\j) at (0.5,0.85+0.13*\j) {};
     \coordinate (z\j) at (2,0.85+0.13*\j) {};
   }
   \foreach \i in {1,...,5}{
     \foreach \j in {1,...,5}{
       \draw (y\i) -- (z\j) ;
     }
   }
   \draw[very thick, red] (1,0.8) -- (1.5,1.7) ;
   \draw[very thick, red] (1,1.7) -- (1.5,0.8) ;
   \node at (0.2,1.24) {$d(k)$} ;
   \node at (2.4,1.24) {$d(k)$} ;
 \end{tikzpicture}
 \caption{Weakly connected, $\min(|A|,|B|) > (d(k)k^2)^{d(k)}$.}
\end{subfigure}
\caption{Summary of the FPT Turing-reductions and their hypotheses, provided we can efficiently find such tripartitions. For (c) and (d), the FPT Turing-reductions are decreasing, while for (a) and (b) they are just improving.}
\label{fig:summaryLemmas}
\end{figure}
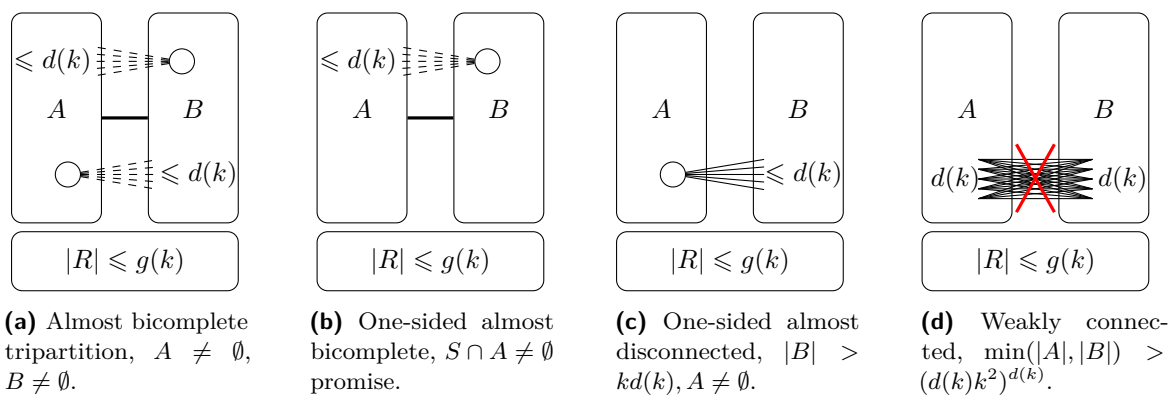

We know provide a few words in order to understand how to use these results.
An obvious caveat is that, even if such a tripartition exists, computing it (or even, approximating it) may not be fixed-parameter tractable.
What we hope is that on a class $\mathcal C$, we will manage to exploit the class structure in order to eventually find such tripartitions, in the cases we cannot conclude by more direct means.  
One of our main results, \cref{thm:p1ttt}, illustrates that mechanism, when the algorithm is centered around getting to the hypotheses of \cref{lem:completeSum1sided} or \cref{thm:withoutBiclique}.

\section{FPT algorithm in \pk{1,t,t,t}-free graphs}\label{sec:p1ttt}

We denote by \pk{a,b,c,d} the graph made by substituting the vertices of $P_4$ by cliques of size $a$, $b$, $c$, and $d$, respectively.
For instance, \pk{1,1,1,2} is $\overline P$ and \pk{1,1,2,1} is the kite.
We settle the parameterized complexity of \mis on $\overline P$-free and kite-free graphs simultaneously (see Figure~\ref{fig:remainingfour}), by showing that \mis is FPT even in the much wider class of \pk{1,t,t,t}-free graphs.

\begin{theorem}\label{thm:p1ttt}
  For every integer $t$, \mis is FPT in \pk{1,t,t,t}-free graphs.
\end{theorem}

\begin{proof}
  Let $t$ be a fixed integer, and $(G,k)$ be an input such that $G$ is \pk{1,t,t,t}-free and $\alpha(G) \leqslant k$.
  We assume that $k \geqslant 3$, otherwise we conclude in polynomial-time.

  The global strategy is the following.
  First we extract a collection $C$ of disjoint and non-adjacent cliques with minimum and maximum size requirements, and some maximality condition.
  Then we partition the remaining vertices into equivalence classes with respect to their neighborhood in $C$.
  The maximum size imposed on the cliques of $C$ ensures that the number of equivalence classes is bounded by a function of $k$.
  Setting $C$ and the small\footnote{the ones whose size is bounded by a later-specified function of $k$} equivalence classes apart, we show that the rest of the graph is partitionable into $(A,B)$ such that either $\mathcal B(A,B)$ is $K_{d(k),d(k)}$-free, in which case we conclude with Theorem~\ref{thm:withoutBiclique}, or $\mathcal B(A,B)$ is almost a complete bipartite graph, in which case we conclude with Lemma~\ref{lem:completeSum1sided} (see Algorithm~\ref{alg:p1ttt} for the pseudo-code).

  As for the running time, we are looking for an algorithm in time $f(k)n^c$ for some fixed constant $c \geqslant 2$, and $f$ an increasing computable function.
  We see $f$ as a partial function on $[k-1]$, and extend it to $[k]$ in the recursive calls.

  \subparagraph*{Building the clique collection $C$.}
  For technical reasons, we want our collection $C$ to contain at least two cliques, at least one of which being fairly large (larger than we can allow ourselves to brute-force).
  So we proceed in the following way.
  We find in polynomial-time $n^{8t+O(1)}$ a $2K_{4t}$.
  If $G$ is $2K_{4t}$-free, an FPT algorithm already exists \cite{BonnetBCTW18}.
  We see these two $K_{4t}$ as the two initial cliques of our collection.
  Let $X$ be the set of vertices with less than $t$ neighbors in at least one of these two $K_{4t}$.
  We partition $X$ into at most $2^{8t}$ vertex-sets (later they will be called \emph{subclasses}) with the same neighborhood on the $2K_{4t}$.
  If all these sets contain less than $\Ramsey(k+1,2kt)$ vertices, $X$ is fairly small: it contains less than $2^{8t}\Ramsey(k+1,2kt)$.
  The other vertices have at least $t$ neighbors in both $K_{4t}$.
  We will show (Lemma~\ref{lem:tImpliesAll}) that this implies that these vertices are completely adjacent to both $K_{4t}$.
  Hence, vertices in the $2K_{4t}$ would have at most $2^{8t}\Ramsey(k+1,2kt)$ non-neighbors.
  In that case, we can safely remove the $2K_{4t}$ from $G$, by Observation~\ref{obs:boundedCoDegree}.
  
  So we can safely assume that (eventually) one subclass of $X$ has more than $\Ramsey(k+1,2kt)$ vertices.
  We can find in polynomial-time a clique $C_2$ of size $2kt$.
  We build a new collection with $3t$ vertices of the first $K_{4t}$, that we name $C_1$.
  We take these vertices not adjacent to $C_2$ (this is possible since vertices in $C_2$ have the same at most $t-1$ neighbors in $K_{4t}$).
  Now we have in $C$ a clique $C_1$ of size $3t$ and a clique $C_2$ of size $2kt$.
  
  We say that a clique of $C$ is \emph{large} if its size is above $kt$, and \emph{small} otherwise.
  We can now specify the requirements on the collection $C$.
  \begin{enumerate}[(1)]
  \item $C$ is a vertex-disjoint and independent\footnote{there is no edge between two cliques of the collection} collection of cliques.
  \item all the cliques have size at least $3t$ and at most $2kt$.
  \item the number of cliques is at least $2$.
  \item if we find a way to strictly increase the number of large cliques in $C$, we do it.
  \end{enumerate}
  
  As $\alpha(G) \leqslant k$, the number of cliques in $C$ cannot exceed $k$.
  This has two positive consequences.
  The first is in conjunction with the way we improve the collection $C$: by always increasing the number of large cliques by $1$.  
  Therefore, we can improve the collection $C$ at most $k-1$ times.
  In particular, the improving process of $C$ terminates (in polynomial time).
  The second benefit is that the total number of vertices of $C$ is always bounded by $2k^2t$.
  Hence, the number of subclasses (sets of vertices with the exact same neighborhood in $C$) is bounded by a function of $k$ (and the constant $t$).

  As a slight abuse of notation, $C_1, \ldots, C_s$ will always be the current collection $C$ $(s<k)$.
  We say that a vertex of $G - C$ \emph{$t$-sees} a clique $C_j$ of $C$ if it has at least $t$ neighbors in $C_j$.
  A \emph{class} is a set of vertices $t$-seeing the same set of cliques of $C$.
  A \emph{subclass} is a a set of vertices with the same neighborhood in $C$.
  Both classes and subclasses partition $G-C$.
  Observe that subclasses naturally refine classes.
  By extension, we say that a (sub)class $t$-sees a clique $C_i \in C$ if one vertex or equivalently all the vertices of that (sub)class $t$-see $C_i$.

  Let $\eta := \lceil (2\Ramsey(k+1,t))^{\Ramsey(k+1,t)}2^{2\Ramsey(k+1,t)-1} \rceil + 1$.
  We choose this value so that $\eta^2/2 > \Ramsey(k+1,t) (2 \eta)^{2-1/\Ramsey(k+1,t)}$ (it will become clear why in the proof of Lemma~\ref{lem:noOverlap}).
  We say that a subclass is \emph{big} if it has more than $\max(\Ramsey(k+1,2kt),\eta)=\eta$ vertices, and \emph{small} otherwise.
  Since $\alpha(G) \leqslant k$, here are two convenient properties on a big subclass:
  \begin{itemize}
  \item a clique of size $t$ can be found in polynomial-time, in order to build a potential \pk{1,t,t,t},
  \item a clique of size $2kt$ can be found, in order to challenge the maximality of $C$.
  \end{itemize}
  We will come back to the significance of $\eta$ later.

  We can now specify item (4) of the clique-collection requirements.
  We resume where we left off the collection $C$, that is $\{C_1=K_{3t}, C_2=K_{2kt}\}$.
  While there is a big subclass that does not $t$-see any large clique of $C$, we find a clique of size $2kt$ in that subclass, and add it to the collection.
  We then remove the small clique ($K_{3t}$) potentially left, and in each large clique of $C$, we remove from $C$ all neighbors of the subclass (they are at most $t-1$ many of them).
  This process adds a large clique to $C$, and decreases the size of the previous large cliques by at most $t-1$.
  Since the large cliques all enter $C$ with size $2kt$, and the number of improvements is smaller than $k$, a large clique will remain large throughout the entire process.
  Therefore, the number of large cliques in $C$ increases by $1$.
  Since we started with one large clique among the first two cliques, the number of cliques remains at least $2$.
  Note that, at each iteration, we update the subclasses with respect to the new collection $C$ (see Algorithm~\ref{alg:cliqueCollection} for the pseudo-code).

  \begin{algorithm}[!h]
  \caption{Routine for computing the clique collection $C$.}
  \label{alg:cliqueCollection}
  \begin{algorithmic}[1]
    \Require{$k$ is a positive integer, $G$ is \emph{not} $2K_{4t}$-free, $\alpha(G) \leqslant k$}
    \Function{BuildCliqueCollection}{$G,k$}:
    \Let{$C$}{$\{K_{4t}, K_{4t}\}$} \Comment{computed by brute-force}
    \If{$\exists$~big subclass not $t$-seeing both $K_{4t}$}
    \Let{$C_2$}{$K_{2kt}$ in the subclass} \Comment{by Ramsey}
    \Let{$C_1$}{$3t$ vertices not adjacent to $C_2$ from one of the $K_{4t}$ not $t$-seen by the subclass}
    \Let{$C$}{$\{C_1,  C_2\}$}    
    \Else{~every big subclasses $t$-see both $K_{4t}$}{\\
      ~~~~~~~~~~~~~~~~~~~~vertices in $C$ have bounded co-degree \Comment{Lemma~\ref{lem:tImpliesAll}}\\
      ~~~~~~~~~~~~~~~~~~~~we can safely delete them \Comment{Observation~\ref{obs:boundedCoDegree}}\\
      ~~~~~~~~~~~~~~~~~~~~and call BuildCliqueCollection($G',k$) with the new graph $G'$
    }
    \EndIf 
    \While{$\exists$~big subclass not $t$-seeing any large clique}
      \Let{$C_j$}{$K_{2kt}$ in the subclass} \Comment{by Ramsey}
      \Let{$C'$}{$C \setminus \{$small clique$\}$} \Comment{this is actually done at most once}
      \Let{$C''$}{map$(C',$deleteNeighborsOf$(C_j))$} \Comment{remove $C_i \cap N(C_j)$ from each $C_i \in C$}
      \Let{$C$}{$C'' \cup \{C_j\}$}    \Comment{the new $C$ contains one more large clique, $C_j$.}
    \EndWhile
      \State \Return{$C$}
      \EndFunction
   \Ensure{output $C$ is a collection of at least two (and at most $k-1$) pairwise independent cliques of size between $3t$ and $2tk$, and every big subclass $t$-sees at least one large clique (\ie, clique of $C$ of size at least $tk$).}
  \end{algorithmic}
\end{algorithm}
  
  Small subclasses are set aside as their size is bounded by a function of $k$.
  Therefore, from hereon, all the subclasses are supposed big.
  We denote by $P(I)$ the class for which $I \subseteq [s]$ represents the indices of the cliques it $t$-sees.
  A first remark is that all the subclasses of $P(\emptyset)$ are small (so we "get rid of" the whole class $P(\emptyset)$).
  \begin{lemma}\label{lem:pnull}
    If $P'$ is a subclass of $P(\emptyset)$, then $|P'| \leqslant \Ramsey(k+1,2kt)$. 
  \end{lemma}
  \begin{proof}
    $P'$ does not $t$-see any (large) clique of $C$.
    So by the maximality property of $C$, it cannot contain a clique of size $2kt$ (see Algorithm~\ref{alg:cliqueCollection}).
    In particular, it cannot have more than $\Ramsey(k+1,2kt)$ vertices.
  \end{proof}

  We turn our attention to classes $P(I)$ with $|I| \geqslant 1$ and their subclasses.
  
  \subparagraph*{Structure of the classes $P(I)$.}
  We show a series of lemmas explaining how classes are connected to $C$ and, more importantly, how they are connected to each other.
  This uses the ability to build cliques of size $t$ at will, in big subclasses.
  Avoiding the formation of \pk{1,t,t,t} will imply relatively dense or relatively sparse connections between classes $P(I)$ and $P(J)$.
  
  \begin{lemma}\label{lem:tImpliesAll}
    If a big subclass $t$-sees at least two cliques $C_i$ and $C_j$ of $C$, then all the vertices of that subclass are adjacent to all the vertices of both cliques.
  \end{lemma}
  \begin{proof}
    We find $D$, a clique of size $t$ in the subclass.
    Let $D_i$ and $D_j$ be $t$ neighbors of the subclass in $C_i$ and $C_j$, respectively.
    Assume that the subclass has a non-neighbor $v \in C_i$.
    Then $vD_iDD_j$ is a \pk{1,t,t,t}.
  \end{proof}
  In light of the previous lemma, if $|I| \geqslant 2$, the cliques of $C$ that the class $P(I)$ $t$-sees are completely adjacent to $P(I)$.

   \begin{lemma}\label{lem:Inclusion}
     Let $I \subsetneq J \subseteq [s]$.
     Then, every vertex of $P(I)$ is adjacent to every vertex of $P(J)$ except at most $\Ramsey(k+1,t)$.
  \end{lemma}
  \begin{proof}
    Let $i \in I$ and $j \in J \setminus I$.
    By Lemma~\ref{lem:tImpliesAll}, all vertices of $P(J)$ are adjacent to all vertices of $C_i \cup C_j$.
    Suppose, by contradiction, that there is a vertex $u \in P(I)$  with more than $\Ramsey(k+1,t)$ non-neighbors in $P(J)$.
    We find a clique $D$ of size $t$ in $G[P(J) \setminus N(u)]$.
    Let $D_i$ be $t$ neighbors of $u$ in $C_i$.
    Let $D_j \subset C_j$ be $t$ neighbors of $P(J)$ which are not neighbors of $u$.
    Such a set $D_j$ necessarily exists since $u$ has at most $t-1$ neighbors in $C_j$, while $P(J)$ is completely adjacent to $C_j$, and $|C_j| \geqslant 3t$.
    Then $uD_iDD_j$ is a \pk{1,t,t,t}.
  \end{proof}

  We say that two sets $I, J$ are \emph{incomparable} if $I$ is not included in $J$, and $J$ is not included in $I$.
  Recall that $\mathcal B(A,B)$ stands for the bipartite graph between vertex-set $A$ and vertex-set~$B$.
  Let $p(t,k) := 2^{2k^2t}$ be a crude upper bound on the total number of subclasses.
   \begin{lemma}\label{lem:disjointClass}
    Let $I,J \subseteq [s]$ be two incomparable sets, and $P_\ell(I), P_{\ell'}(J)$ be any pair of subclasses of $P(I)$ and $P(J)$, respectively.
    Then, $\mathcal B(P_\ell(I),P_{\ell'}(J))$ is $K_{\Ramsey(k+1,t),\Ramsey(k+1,t)}$-free.
    Hence, $\mathcal B(P(I),P(J))$ is $K_{p(t,k)\Ramsey(k+1,t),p(t,k)\Ramsey(k+1,t)}$-free. 
   \end{lemma}

   \begin{proof}
     Let $i \in I \setminus J$ and $j \in J \setminus I$.
     We first assume that one of $I, J$, say $I$, has at least two elements.
     Suppose, by contradiction, that there is a set $B_I \subseteq P_\ell(I)$ and a set $B_J \subseteq P_{\ell'}(J)$ both of size $\Ramsey(k+1,t)$, such that there is no non-edge between $B_I$ and $B_J$.
     Let $u$ be a vertex of $C_j$ which is adjacent to $P_{\ell'}(J)$ but not to $P_\ell(I)$.
     We find $D_I$, a clique of size $t$ in $G[B_I]$, and $D_J$, a clique of size $t$ in $G[B_J]$.
     Let $D_i$ be $t$ neighbors of $P_\ell(I)$ in $C_i$ that are not adjacent to $P_{\ell'}(J)$.
     Those $t$ vertices exist since, by Lemma~\ref{lem:tImpliesAll}, $P_\ell(I)$ is completely adjacent to $C_i$ (by assumption $|I| \geqslant 2$).
     And $P_{\ell'}(J)$ has more than $t$ non-neighbors in $C_i$.
     Then, $uD_JD_ID_i$ is a \pk{1,t,t,t}.

     We now have to settle the remaining case: $|I|=|J|=1$ ($I=\{i\}$ and $J=\{j\}$).
     If $P_\ell(I)$ has at least $2t$ neighbors in $C_i$ or $P_{\ell'}(J)$ has at least $2t$ neighbors in $C_j$, we conclude as in the previous paragraph.
     So we assume that it is not the case.
     We distinguish two cases.
     
     Either $P_\ell(I)$ has at least one neighbor in $C_j$, say $u$.
     Let $D_I$ be a clique of size $t$ in $P_\ell(I)$, $D_i \subseteq C_i$ be $t$ neighbors of $P_\ell(I)$, and $D'_i \subseteq C_i$ be $t$ non-neighbors of $P_\ell(I)$.
     $D_i$ and $D'_i$ exist since $P_\ell(I)$ has between $t$ and $2t-1$ neighbors in $C_j$, and $|C_j| \geqslant 3t$.
     Then, $uD_ID_iD'_i$ is a \pk{1,t,t,t}.
     
     Or $P_\ell(I)$ has no neighbor in $C_j$.
     Let $u$ be a non-neighbor of $P_{\ell'}(J)$ in $C_j$, and $D_j \subseteq C_j$ be $t$ neighbors of $P_{\ell'}(J)$.
     If there is a set $B_I \subseteq P_\ell(I)$ and a set $B_J \subseteq P_{\ell'}(J)$ both of size $\Ramsey(k+1,t)$, such that $B_I$ and $B_J$ are completely adjacent to each other.
     We can find $D_I$, a clique of size $t$ in $G[B_I]$, and $D_J$, a clique of size $t$ in $G[B_J]$.
     Then, $uD_jD_JD_I$ is a \pk{1,t,t,t}.
     This implies that, in any case, there cannot be a $K_{p(t,k)\Ramsey(k+1,t),p(t,k)\Ramsey(k+1,t)}$ in $\mathcal B(P(I),P(J))$.
   \end{proof}

   We say that the sets $I$ and $J$ \emph{overlap} if all three of $I \cap J$, $I \setminus J$, $J \setminus I$ are non-empty.
   \begin{lemma}\label{lem:noOverlap}
     Let $I,J \subseteq [s]$ be two overlapping sets.
     Then, at least one of $P(I)$ and $P(J)$ have only small subclasses.
   \end{lemma}
   \begin{proof}
     Suppose, by contradiction, that there is a big subclass $P_\ell(I)$ of $P(I)$, and a big subclass $P_{\ell'}(J)$ of $P(J)$.
     Observe that, for $I$ and $J$ to overlap, their size should be at least $2$.
     Let $i \in I \setminus J$, $j \in J \setminus I$, $h \in I \cap J$.
     By the arguments of Lemma~\ref{lem:Inclusion} applied to the restriction to $P(I)$, $P(J)$, $C_h$, and $C_j$, a vertex in $P(I)$ has at most $\Ramsey(k+1,t)$ non-neighbors in $P(J)$.
     Let us consider $\eta$ vertices in $P_\ell(I)$ and $\eta$ vertices in $P_{\ell'}(J)$.
     Since $\eta \geqslant 2\Ramsey(k+1,t)$, the previous observation implies that the number of edges between them is at least $\eta^2/2$. 
     But by Lemma~\ref{lem:disjointClass}, the bipartite graph linking them should be $K_{\Ramsey(k+1,t),\Ramsey(k+1,t)}$-free.
     By K\H{o}v\'ari-S\'os-Tur\'an's theorem, this number of edges is bounded from above by $\Ramsey(k+1,t) (2 \eta)^{2-1/\Ramsey(k+1,t)} < \eta^2/2$, a contradiction.
   \end{proof}
   
   Hence, the remaining (not entirely made of small subclasses) classes define a laminar\footnote{where two sets are nested or disjoint} set-system.
   We denote by $R$ the union of the vertices in all the small subclasses and $C$.
   We now add a new condition to be a small subclass (condition that we did not need thus far).
   A subclass is also small if it has at most $|R|$ vertices.
   Note that this condition can snowball.
   But eventually $R$ has size bounded by $g(k) := 2^{p(t,k)}(p(t,k)\eta + 2k^2t)$.
   A class is \emph{remaining} if it contains at least one big subclass.
   By Lemma~\ref{lem:pnull}, $P(\emptyset)$ cannot be remaining.
   If no class is remaining, then the whole graph is a kernel.
   So we can assume that there is at least one remaining class.
   Let $P(I)$ be a remaining class in $G-R$ such that $I$ is maximal among the remaining classes.
   We distinguish two cases: either there is at least one other remaining class $P(J)$ ($I \neq J$), or $P(I)$ is the unique remaining class.
   
   \subparagraph*{At least two remaining classes $P(I)$ and $P(J)$.}
   By Lemma~\ref{lem:noOverlap}, any other class $P(J)$ satisfies $J \subsetneq I$ or $I \cap J = \emptyset$.
   Let $\iota, \delta \leqslant 2^k$ be the number of remaining classes such that $J \subsetneq I$ and such that $I \cap J = \emptyset$, respectively.
   Again, we distinguish two cases: $\delta > 0$, and $\delta = 0$.
   If $\delta > 0$, we build the partition $(A,B,R)$ of $V(G)$ such that $A$ contains the $\iota+1$ classes whose set is included in $I$ and $B$ contains the $\delta$ classes whose set is disjoint from $I$.
   By Lemma~\ref{lem:disjointClass}, the bipartite graph between any of the $(\iota+1)\delta$ pairs of classes made of one class whose set is contained in $I$ and one class whose set is disjoint from $I$ is $K_{p(t,k)\Ramsey(k+1,t),p(t,k)\Ramsey(k+1,t)}$-free.
   Hence, the bipartite graph between $A$ and $B$ is $K_{2^kp(t,k)\Ramsey(k+1,t),2^kp(t,k)\Ramsey(k+1,t)}$-free.
   Thus we conclude by Theorem~\ref{thm:withoutBiclique} with $d(k)=2^kp(t,k)\Ramsey(k+1,t)$.

   We now tackle the case $\delta=0$, that is, all the remaining classes $P(J)$ satisfy $J \subseteq I$.
   We first assume that there are two remaining classes with disjoint sets.
   A laminar set-system with a unique maximal set can be represented as a rooted tree, where nodes are in one-to-one correspondence with the sets, and the parent-to-child arrow represents the partial order of inclusion.
   Here, the root is labeled by $I$ (since $I$ is the unique maximal set), and all the nodes are labeled by a subset of $[s]$ corresponding to a remaining class.
   Let $I=I_1 \supsetneq I_2 \supsetneq \ldots \supsetneq I_h$ be the path from the root to the first node with out-degree at least $2$.
   Observe that $C$ contains at most $k$ cliques, so $h \leqslant k$.
   Let $J_1, J_2, \ldots, J_\ell$ be the $\ell$ children of $I_h$ (with $\ell \geqslant 2$).
   Let $\mathcal P_1$ be the remaining classes whose set is included in $J_1$, and $\mathcal P_{2+}$ be the remaining classes whose set is included in one $J_i$ for some $i \in [2,\ell]$.
   Let $A := \bigcup_{q \in [h]} P(I_q)$, and $B := V(G) \setminus (A \cup R)$.
   By Lemma~\ref{lem:Inclusion}, vertices of $B$ have at most $h\Ramsey(k+1,t) \leqslant k \Ramsey(k+1,t)$ non-neighbors in $A$.
   We apply Lemma~\ref{lem:completeSum1sided} with the tripartition $(A,B,R)$ and $d_1(k) = k \Ramsey(k+1,t)$.
   Only we did not cover the case in which the solution does not intersect $A$.
   We do so by applying Theorem~\ref{thm:withoutBiclique} to the tripartition $(\mathcal P_1,\mathcal P_{2+},R)$ with $d_2(k) = 2^kp(t,k)\Ramsey(k+1,t)$.
   A priori, what we just did is not bounded by $f(k)|V(G)|^c$, hence not legal.
   Let us go back to the last lines of Lemma~\ref{lem:completeSum1sided} and of Theorem~\ref{thm:withoutBiclique}.
   Our running time is bounded by $f(k)|A \cup R|^c + k^2{d_1(k) \choose k}d_1(k)^cf(k-1)|B|^c + k(k+2)(\lceil d_2(k)^{d_2(k)} k^{2d_2(k)-1} \rceil +1)f(k-1)|B|^c$, where the two first terms come from the application of Lemma~\ref{lem:completeSum1sided}, and the third term, from Theorem~\ref{thm:withoutBiclique}.
   This is at most $f(k)|A \cup R|^c+f(k)|B|^c \leqslant f(k)|V(G)|^c$ by Cauchy-Schwarz inequality, with $f(k) := (k^2{d_1(k) \choose k}d_1(k)^c + k(k+2)(\lceil d_2(k)^{d_2(k)} k^{2d_2(k)-1} \rceil +1))f(k-1)$.

   Let now assume that all the remaining classes have nested sets (no two sets are disjoint).
   Let $I=I_1 \supsetneq I_2 \supsetneq \ldots \supsetneq I_h$ the sets of \emph{all} the remaining classes ($h \leqslant k$).
   Suppose $h \geqslant 3$.
   We apply Lemma~\ref{lem:completeSum1sided} to the tripartition $(P(I_1) \cup P(I_2), \bigcup_{j \in [3,h]}P(I_j),R)$ with $d(k) = 2 \Ramsey(k+1,t)$.
   Indeed, by Lemma~\ref{lem:Inclusion}, vertices of $\bigcup_{j \in [3,h]}P(I_j)$ have at most $\Ramsey(k+1,t)$ non-neighbors in $P(I_1)$ and at most $\Ramsey(k+1,t)$ non-neighbors in $P(I_2)$.
   We deal with the case in which the solution does not intersect $P(I_1) \cup P(I_2)$ in the following way.
   Let $C_q$ be the clique of $C$ only $t$-seen by $P(I_1)$ and $C_{q'}$ the clique of $C$ only $t$-seen by $P(I_1) \cup P(I_2)$.
   One of these two cliques has to be large (since there is at most one small clique).
   We branch on the at least $tk$ and at most $2tk$ vertices of that large clique, say $C'$.
   A maximal independent set cannot be fully contained in $\bigcup_{j \in [3,h]}P(I_j)$.
   Indeed, any choice of at most $k$ vertices in this set dominates at most $k(t-1)$ vertices of $C'$.
   Thus, we cannot miss a solution.
   Let us turn to the running time.
   Once again, we cannot use Lemma~\ref{lem:completeSum1sided} as a total black-box.
   Our running time is bounded by $f(k)|A \cup R|^c + k^2{d(k) \choose k}d(k)^c f(k-1)|B|^c + 2tkf(k-1)|B \cup R|^c \leqslant f(k)|A \cup R|^c + f(k)|B \cup R|^c$ with $f(k) := (k^2{d(k) \choose k}d(k)^c+2tkf)f(k-1)$, and $f(k)|A \cup R|^c + f(k)|B \cup R|^c \leqslant f(k)|V(G)|^c$, by Lemma~\ref{lem:almostSplit}.
   Here we need that $|A|>|R|$ and $|B|>|R|$ which is the case: recall that we added that requirement to be a big subclass.

   The last case is the following.
   There are exactly two remaining classes associated to sets $I=I_1 \supsetneq I_2$.
   If a clique not $t$-seen by $P(I_2)$ is large or if $P(I_2)$ is $2K_{4t}$-free, we conclude with Lemma~\ref{lem:completeSum1sided} (recall that this finds a solution if there is one intersecting $P(I_1)$.
   In both cases, if the solution does not intersect $P(I_1)$, we can find it with only a small overhead cost.
   If a clique not $t$-seen by $P(I_2)$ is large, we branch on the at most $2kt$ vertices of that clique.
   If $P(I_2)$ is $2K_{4t}$-free, an independent set of size $k$ can be found in $G[P(I_2)]$ in FPT time \cite{BonnetBCTW18}.

   Finally, we can assume that $G[P(I_2)]$ contains a $2K_{4t,4t}$ and does not $t$-see a small clique in $C$.
   Note that this implies that $C$ is made of two cliques $K_{3t}$ and $K_{2kt}$.
   We call \emph{critical} such a case where $C=\{K_{3t},K_{2kt}\}$ and a $2K_{4t}$ can be found in a class not $t$-seeing $K_{3t}$. 
      
   For this very specific case (that may also arise with a unique remaining class, see below), we perform the following refinement of the clique-collection computation.
   We compute a new clique collection, say $C^2$, in $G-C$, starting with a $2K_{4t,4t}$ found in the class not $t$-seeing the previous $K_{3t}$. 
   If $C^2$ is not of the form $\{K_{3t},K_{2kt}\}$, we add $C$ to the bounded-in-$k$ set $R$, and we follow our algorithm (that is, a non-critical case).
   If $C^2=\{K_{3t},K_{2kt}\}$, we compute a new clique collection $C^3$ in $G-(C^1 \cup C^2)$ (with $C^1=C$), again starting with a $2K_{4t,4t}$ found in the class not $t$-seeing the previous $K_{3t}$, and so on.
   Let us assume that we are always in a critical case, with $C^h=\{C^h_1=K_{3t},C^h_2=K_{2kt}\}$.
   We stop after $\zeta := \Ramsey_{2^{(3t)^2}}(4kt)$ iterations, leading to disjoint (though not independent) clique collections $C=C^1, C^2, \ldots, C^\zeta$.
   In particular, $|\bigcup_{h \in \zeta} C^h|$ is still bounded by a function of $k$, namely $\zeta(3t+2kt)$.
   We claim that we can find a $2K_{2kt,2kt}$ in $G[\bigcup_{h \in \zeta} C^h_1]$.

   Because of the number of iterations, one can extract $4kt$ cliques $C^h_1$ (of size $3t$) with the same bipartite graph linking any pair of $C^h_1$ (with a fixed but arbitrary ordering of each $C^h_1$).
   This common bipartite graph has to be empty, complete, or a half-graph.
   Let us show that it can only be a half-graph.
   For any $i \in [3t]$, the $i$-th vertices in the $C^h_1$ should be adjacent (otherwise they form an independent set of size $2kt$).
   That excludes the empty bipartite graph.
   Let $h_1$ be the smallest index such that we have extracted $C^{h_1}_1$.
   The common bipartite graph cannot be complete either, since all the vertices of $G-(\bigcup_{h \in [h_1]})$ have at most $t-1$ neighbors in $C^{h_1}_1$.
   This was one of the condition of a critical case.
   So the bipartite graph is a half-graph.
   Then we find our $2K_{2kt,2kt}$ as the first vertex (or last vertex) of the first $2kt$ extracted cliques, and the last vertex (or first vertex) of the last $2kt$ extracted cliques.
   Now we finally have a clique collection with two independent \emph{large} cliques, depending on the orientation of the half-graph.
   So we can start again without reaching the problematic case.
   
   \subparagraph*{Unique remaining $P(I)$.}
   If $|I| \geqslant 2$, by Lemma~\ref{lem:tImpliesAll}, $P(I)$ is completely adjacent to one clique $C_i$ (with $i \in I$).
   Any vertex of $C_i$ has at most $g(k)$ non-neighbors.
   This case is handled by Observation~\ref{obs:boundedCoDegree}.
   So we now suppose that $|I|=1$ (and $I=\{i\}$).
   If $P(I)$ does not $t$-see a large clique $C_j$, we can branch on the at most $2kt$ vertices of that clique.
   Indeed, there is a solution that intersects it, since $k-1$ vertices in $G-R$ can dominate at most $(k-1)(t-1)<kt$ vertices.
   Thus, we can further assume that $P(I)$ $t$-sees all the large cliques.
   This forces that there is at most one large clique, since $|I|=1$.
   There cannot be at least three cliques in $C$.
   Indeed, the way the collection is maintained, that would imply that there are at least two large cliques.
   So, $C=\{C_1=K_{3t}, C_2=K_{2kt}\}$ and $I=\{2\}$.
   This is a \emph{critical} case, which we handle as in the previous paragraph (with two remaining classes).
     \begin{algorithm}[h!]
  \caption{FPT algorithm for \mis on \pk{1,t,t,t}-free graphs}
  \label{alg:p1ttt}
  \begin{algorithmic}[1]
    \Require{$G$ is \pk{1,t,t,t}-free, $k \geqslant \alpha(G)$}
    \Function{Stable}{$G,k$}:
    \If{$k \leqslant 2$} solve in $n^2$ by brute-force \EndIf \Comment{now $k \geqslant 3$}
    \If{$G$ is $2K_{4t}$-free} solve in FPT time \EndIf \Comment{see~\cite{BonnetBCTW18}}
    \Let{$C$}{BuildCliqueCollection$(G,k)$}
    \Let{$R$}{$C~\cup~$ subclasses of size less than $\eta$}   \Comment{small subclasses are set aside}
    \While{$\exists$~subclass $Q$ of size at most $|R|$}
        \Let{$R$}{$R \cup Q$}
    \EndWhile
    \Let{$\mathcal P$}{remaining classes}
    \If{$\mathcal P=\emptyset$} input is a kernel \EndIf  
    \Let{$P(I)$}{remaining class with $I$ maximal for inclusion}
    \If{$|\mathcal P| \geqslant 2$}{
      \If{$\exists P(J) \in \mathcal P$ such that $I \cap J = \emptyset$}\\
      ~~~~~~~~~~~~~~~$(A,B,R)$ with $\mathcal B(A,B)$ $K_{d(k),d(k)}$-free \Comment{Theorem~\ref{thm:withoutBiclique}}
      \EndIf
      \If{$\forall P(J) \in \mathcal P$, $J \subseteq I$}\\
      ~~~~~~~~~~~~~~~$(A,B,R)$ with $\forall v \in B$, $v$ has co-degree $\leqslant d_1(k)$ in $A$  \Comment{Lemma~\ref{lem:completeSum1sided}}\\
      ~~~~~~~~~~~~~~~and $(B_1,B_2,R)$ in $G[B \cup R]$ with $\mathcal B(B_1,B_2)$ $K_{d_2(k),d_2(k)}$-free, \Comment{Theorem~\ref{thm:withoutBiclique}}\\
      ~~~~~~~~~~~~~~~or branching on $2tk$ vertices,\\
      ~~~~~~~~~~~~~~~or critical case, when repeated, yields a $2K_{2kt,2kt}$
      \EndIf
    }
    \EndIf
        \If{$\mathcal P =\{P(I)\}$}{
    a vertex of $C$ has small co-degree,                  \Comment{see Observation~\ref{obs:boundedCoDegree}}\\
    ~~~~~~~~~~~~~~~~~~~~~~~~~~~~~~~~or branching on $2tk$ vertices, \\
    ~~~~~~~~~~~~~~~~~~~~~~~~~~~~~~~~or critical case, when repeated, yields a $2K_{2kt,2kt}$
    }
    \EndIf
      \EndFunction
  \end{algorithmic}
\end{algorithm}
\end{proof}

\section{Randomized FPT algorithms in dart-free and cricket-free graphs}\label{sec:dartAndCricket}
In this section, we consider the case of dart-free and cricket-free graphs, and prove that there is a randomized FPT algorithm for \mis in both graph classes.
To this end, we use the technique of iterative expansion together with a Ramsey extraction.

Let us first define the Ramsey extraction. Notice that for all results presented in this paper, the cliques will be of size $1$ or $2$ only, which greatly simplifies the construction. However, we give the complete version of the definition, as introduced in~\cite{BonnetBCTW18}.

\begin{definition}
Given a graph $G$ and a set of $k-1$ vertex-disjoint cliques of $G$, $\C = \{C_1, \dots, C_{k-1}\}$, each of size $q$, we say that $\C$ is a set of \emph{Ramsey-extracted cliques of size $q$} if the conditions below hold. Let $C_r = \{c_j^r : j \in \{1, \dots, q\}\}$ for every $r \in \{1, \dots, k-1\}$.
\begin{itemize}
	\item For every $j \in [q]$, the set $\{c_j^r : r \in \{1, \dots, k-1\}\}$ is an independent set of $G$ of size $k-1$.
	\item For any $r \neq r' \in \{1, \dots, k-1\}$, one of the four following case can happen:
	\begin{enumerate}[(i)]
		\item for every $j,j' \in [q]$, $c_j^rc_{j'}^{r'} \notin E(G)$ 
		\item for every $j,j' \in [q]$, $c_j^rc_{j'}^{r'} \in E(G)$ iff $j \neq j'$
		\item for every $j,j' \in [q]$, $c_j^rc_{j'}^{r'} \in E(G)$ iff $j < j'$
		\item for every $j,j' \in [q]$, $c_j^rc_{j'}^{r'} \in E(G)$ iff $j > j'$
	\end{enumerate}
	In the case $(i)$ (resp. $(ii)$), we say that the relation between $C_r$ and $C_{r'}$ is \emph{empty} (resp. \emph{full}\footnote{Remark that in this case, the graph induced by $C_r \cup C_{r'}$ is the complement of a perfect matching.}). In case $(iii)$ or $(iv)$, we say the relation is \emph{semi-full}.
\end{itemize}
\end{definition}

We then define the following problem:

\begin{definition}\label{def:faugramsey}
The \faugramsey{$f$} problem takes as input an integer $k$ and a graph $G$ whose vertices are partitioned into non-empty sets $X_1 \cup \dots \cup X_k \cup C_1 \cup \dots \cup C_{k-1}$, where:
\begin{itemize}
	\item  $\{C_1, \dots, C_{k-1}\}$ is a set of $k-1$ Ramsey-extracted cliques of size $f(k)$
	\item any independent set of size $k$ in $G$ is contained in $X_1 \cup \dots \cup X_k$ 
	\item $\forall i \in \{1, \dots, k\}$, $\forall v,w \in X_i$ and $\forall j \in \{1, \dots, k-1\}$, $N(v) \cap C_j=N(w) \cap C_j=\emptyset$ or $N(v) \cap C_j =N(w) \cap C_j=C_j$
	\item the following bipartite graph $\B$ is connected: $V(\B) = B_1 \cup B_2$, $B_1=\{b_1^1, \dots, b_k^1\}$, $B_2=\{b_1^2, \dots, b_{k-1}^2\}$ and $b_j^1b_r^2 \in E(\B)$ iff $X_j$ and $C_r$ are adjacent.
	\end{itemize}
The objective is the following:
\begin{itemize}
	\item if $G$ contains an independent set $S$ such that $S \cap X_i \neq \emptyset$ for all $i \in \{1, \dots, k\}$, then the algorithm must answer ``YES''. In that case the solution is called a \emph{rainbow independent set}.
	\item if $G$ does not contain an independent set of size $k$, then the algorithm must answer ``NO''.
\end{itemize}
\end{definition}

As well as the following result from~\cite{BonnetBCTW18}, which allows to focus on \faugramsey{$f$}.

\begin{theorem}{\cite{BonnetBCTW18}}\label{thm:iterative-ramsey}
Let $\G$ be a hereditary graph class. If \faugramsey{$f$} is $FPT$ in $\G$ for some computable function $f$, then \mis is $FPT$ in $\G$.
\end{theorem}

Hence, in the following, by \emph{positive instance} we mean an instance having a rainbow independent set, while \emph{negative instance} denotes instances not containing an independent set of size $k$. If the input graph contains at least one independent set of size $k$, but none of them is rainbow, then we are allowed to output ``No'' (this case is handled by the color coding technique in the proof of the previous result).

For both dart-free and cricket-free graphs, our algorithms will consist of a bounded search tree solving \faugramsey{$f$}. Let us first describe some common ingredients of both algorithms. Throughout the proof, $I = (G, k)$ denote our input of \faugramsey{$f$}. 

 The \emph{set-graph} $R_G$ is the graph on $k$ vertices obtained by contracting all $X_i$'s to single vertices, putting an edge between $X_i$ and $X_j$ if there exists an edge between these two sets.

Our proof consists of a bounded search tree. To this end, we will rely on branchings, which will decrease the parameter $\kappa(G, k)$ defined in the following lexicographic order:
\begin{enumerate}
	\item $k$
	\item $|E(R_G)|$
	\item $\sum_{i=1}^k \alpha(G[X_i])$
\end{enumerate} 

Notice that the number of possible values of $\kappa(G, k)$ is bounded by a function of $k$ only. The idea of our algorithm is to restrict more and more the structure of an instance by the mean of branchings. To this end, our branchings will either strictly decrease our parameter $\kappa(I)$, or create instances with particular properties. This is formalized by the following definition.

\begin{definition}
Given an instance $I$, we say that we can FPT-reduce (or simply \emph{reduce}) to some set of instances $\mathcal{I}$ if there is a computable function $g : \mathbb{N} \rightarrow \mathbb{N}$ such that we can output a set of instances $I_1$, $\dots$, $I_{g(k)}$ such that for every $p \in [g(k)]$, either $\kappa(I_p) < \kappa(I)$, or $I_p \in \mathcal{I}$. This step must run in time polynomial in $\sum_{p=1}^{g(k)} |I_p|$.
\end{definition}

At some places, we will sometimes make use of \emph{random reductions}: the instance $I$ will be transformed into another instance $I'$ with a random rule. We will then ensure that $\kappa(I') < \kappa(I)$, and:
\begin{itemize}
	\item if $I$ is a positive instance (the graph admits a rainbow independent set), then $I'$ is a positive instance with probability at least $h(k)$ for some computable function $h$
	\item if $I$ is a negative instance, then $I'$ is a negative instance.
\end{itemize}
It is easy to see that since the parameter $\kappa$ decreases, the resulting algorithm is a one-sided error Monte Carlo algorithm with a success probability depending on $k$ only. Hence, both algorithms will be \emph{randomized FPT algorithms}.

We begin with a branching whose purpose is to ``clean'' some carefully chosen adjacencies of the set graph.

\begin{definition}
We say that a couple $(i, j) \in [k]^2$ is \emph{clean} if the following conditions are satisfied:
\begin{itemize}
	\item for every $x \in X_i$, $x$ has a non-neighbor in $X_j$ \label{item:nonneighbor}
	\item for every $x \in X_i$, $x$ has a neighbor in $X_j$ \label{item:neighbor}
	\item $G[X_i]$ is connected. \label{item:connected} 
\end{itemize}
\end{definition}

Observe that if a couple $(i, j)$ is clean, then there exists $xy$ in $E(G[X_i])$ and $z \in X_j$ such that $xz \in E(G)$ but $yz \notin E(G)$. We say that a set of couples $\mathcal{P} \subseteq [k]^2$ is \emph{acyclic} if the oriented graph whose vertices are $[k]$ and arcs are $\mathcal{P}$ is acyclic.

\begin{lemma}\label{lemma:clean}
Given an acyclic set of couples $\mathcal{P} \subseteq [k]^2$ corresponding to some edges of the set graph, one can FPT-reduce to instances where every couple $(i, j) \in \mathcal{P}$ is clean.
\end{lemma}
\begin{proof}
Let $\mathcal{P} = \{p_1, \dots, p_t\}$ considered in a total ordering (recall that $\mathcal{P}$ is acyclic). Iteratively for $\ell$ from $t$ downto $1$, we FPT-reduce to an instance where couples $p_{\ell}$, $\dots$, $p_t$ are clean. Let $p_{\ell} = (i, j)$. If $X_i$ and $X_j$ are already clean, we are done. Otherwise, let $U \subseteq X_i$ be the vertices $x$ of $X_i$ such that $xz \in E(G)$ for every $z \in X_j$, let $S \subseteq X_i$ be the vertices $x$ of $X_i$ such that $xz \notin E(G)$ for every $z \in X_j$, and let $\Omega_1$, $\dots$, $\Omega_q$ be the connected components of $G[X_i \setminus (U \cup S)]$ (we may assume $q < k$, otherwise there is an independent set of size $k$ in $G$). We output the following new instances:
\begin{itemize}
	\item $I_S$, where $X_i$ is replaced by $S$
	\item $I_{\Omega_r}$, where $X_i$ is replaced by $\Omega_r$, for every $r \in [q]$.
\end{itemize}
In every new instance, the graph is an induced subgraph of the former one, hence, if $I$ is negative, all new instances are negative as well.
Moreover, if there is a rainbow independent set in $I$, it must intersect either $S$ or $\Omega_r$ for some $r \in [q]$ (it cannot intersect $U$), hence if $I$ is positive then one of the new instances is positive.
Now, observe that $\kappa(I_S) < \kappa(I)$ since there is no edge between $S$ and $X_j$ (while there was an edge between $X_i$ and $X_j$).
Moreover, if $q > 1$, then $\kappa(I_{\Omega_r}) < \kappa(I)$ for every $r \in [q]$, since the size of a maximum independent set in $G[\omega_r]$ has decreased.
Finally, if $q=1$, in $I_{\Omega_1}$, the couples $p_{\ell}$, $\dots$, $p_t$ are clean, since we only modified $X_i$, and $i$ does not appear in $p_s$ for every $s > \ell$ (recall that we consider couples of $\mathcal{P}$ in an inverted total ordering).
\end{proof}

One simple case is where the set graph is a cycle or a path:

\begin{lemma}[Particular set graph: cycle or path]\label{lem:pathcycle}
If the set graph is a path or a cycle, then \faugramsey{$f$} is polynomial-time solvable for every computable function $f$.
\end{lemma}
\begin{proof}
Assume the set graph is a path, with edges between $X_i$ and $X_{i+1}$, $i \in [k-1]$. In that case, observe that there is a solution if and only if the following dynamic programming returns $true$ on input $P(2, x_1)$ for some $x_1 \in X_1$:

$$
P(i,x_{i-1}) = \left\{
    \begin{array}{ll}
    	true & \mbox{ if } i=k \\
        false & \mbox{ if } X_i \subseteq   N(x_{i-1}) \\
        \bigvee_{x_i \in X_i \setminus N(x_{i-1})} P(i+1, x_i) & \mbox{ otherwise.}
    \end{array}
\right.
$$

Clearly this dynamic programming runs in $O(mnk)$ time, where $m$ and $n$ are the number of edges and vertices of the graph induced by $\cup_{i=1}^k X_i$, respectively. Similar ideas can be used when the set graph is a cycle.
\end{proof}

We are now ready to present our algorithms for dart-free and cricket-free graphs.

\subsection{The dart}

\begin{theorem}
There is a randomized FPT algorithm for \mis in dart-free graphs parameterized by the size of the solution.
\end{theorem}
\begin{proof}
As said previously, we solve \faugramsey{$f$} in dart-free graphs. Here, we define $f(x) = 1$ for every $x \in \mathbb{N}$. Hence, for every $j \in [k-1]$, $C_j = \{c_j\}$. The strategy is to use FPT branching and random reductions in order to simplify the structure of the set graph. More precisely, the goal is to reduce to the case where the set graph is $\{$paw, claw$\}$-free\footnote{Recall that the paw is the graph obtained by adding one edge to the claw.}.
  
\begin{lemma}[Removal of some $P_3$'s: part I]\label{lem:P3}
One can FPT-reduce to instances where for every triple $(X_1, X_2, X_3)$ of the set graph inducing a $P_3$, no vertex of $\{c_1, \dots, c_{k-1}\}$ is adjacent to $X_1$, $X_2$ and $X_3$.
\end{lemma}
\begin{proof}
Suppose that this is the case, and let $c_j$ be such a vertex. Then apply the branching of Lemma~\ref{lemma:clean} with couples $\{(1, 2), (2, 3)\}$. We end up with an instance where there exists an edge $xy$ induced by $X_1$ and a vertex $z \in X_2$ such that $xz \in E(G)$ but $yz \notin E(G)$, and a vertex $u \in X_3$ such that $zu \notin E(G)$. But observe that $\{c_j, x, y, z, u\}$ induces a dart, which is impossible.
\end{proof}

\begin{lemma}[Removal of some $P_3$'s: part II]\label{lem:P3bis}
One can FPT-reduce to instances where for every triple $(X_1, X_2, X_3)$ of the set graph inducing a $P_3$, no vertex of $\{c_1, \dots, c_{k-1}\}$ is adjacent to $X_2$ but not to $X_1$ nor $X_3$.
\end{lemma}
\begin{proof}
Suppose that this is the case: let $c_j$ be such a vertex. Apply Lemma~\ref{lemma:clean} with couples $\{(2, 1), (2, 3)\}$. We end up with an instance where every vertex $x \in X_2$ has a neighbor in $X_1$ and $X_3$. But observe that for every edge $xy$ induced by $X_2$, $x$ and $y$ must be a module with respect to $X_1$ and $X_3$. Indeed, if there is a vertex $u \in X_1$ both adjacent to $x$ and $y$, then for every neighbor $z \in X_3$ of $x$, $z$ must also be a neighbor of $y$, for otherwise $\{c_j, x, y, u, z\}$ would induce a dart. But then it means that $z \in X_3$ is adjacent to both $x$ and $y$, and for the same reasons, the neighborhood of $x$ and $y$ in $X_1$ must be the same. 
Hence, we can partition $X_2$ into subsets $\{M_i\}_{i=1}^{p}$, each of which being both a maximal module with respect to $X_1$ and $X_3$. 
For every $i \in [p]$, denote by $N_i$ the set of common neighbors of $M_i$ in $X_1$. Note that the previous paragraph ensures that if there is an edge between $M_i$ and $M_j$, then $N_i$ and $N_j$ must be disjoint. 
Now let us consider any vertex $x \in X_1$. If $x$ is adjacent to both a vertex of $M_i$ and $M_j$ (with $i \neq j$), then, since $x$ is complete to these sets, it means that there is no edge between $M_i$ and $M_j$ since $N_i \cap N_j \neq \emptyset$. In particular, $x$ is only adjacent to at most $k-1$ modules $M_i$ (since otherwise we would be able to find in polynomial time an independent set of size $k$ in $X_2$).

We now proceed to a random reduction: for every $i \in [p]$, we delete all vertices of $M_i$ with probability $1/2$, and after that, we remove every vertex $x \in X_1$ if it has a remaining neighbor in $X_2$. Assume the instance is positive: let $s_1 \in X_1$ and $s_2 \in X_2$ be the elements of a rainbow independent set. Observe that the probability that $s_2$ has not been removed is at least $1/2$, while the probability that $s_1$ has not been removed is at least $1/2^{k-1}$ (since $s_1$ is adjacent to at most $k-1$ modules, and $s_1$ is kept if all its possible neighbors have been deleted). Hence, after this removal step, the obtained instance is positive with probability at least $1/2^k$ (and if the instance was negative, the reduced one shall be negative as well, since the reduced graph is an induced subgraph of the former). Moreover, in the reduced instance there is no edge between $X_1$ and $X_2$, hence the parameter $\kappa$ decreases.
\end{proof}

\begin{lemma}[No claw in the set graph]
One can FPT-reduce to instances where the set graph is claw-free.
\end{lemma}
\begin{proof}
We apply Lemmas~\ref{lem:P3} and \ref{lem:P3bis}, and claim that we end up with instances where the set graph is claw-free. Suppose it is not: let $X_1$ be the center of the claw, and let $c_j \in \{c_1, \dots, c_{k-1}\}$ adjacent to $X_1$. There are two cases: $c_j$ is also adjacent to at least two neighbors of $X_1$, but this is impossible by Lemma~\ref{lem:P3}. The other case is also impossible by Lemma~\ref{lem:P3bis}.
\end{proof}


\begin{lemma}[Removal of paws in the set graph]
One can FPT-reduce to instances where the set graph is paw-free.
\end{lemma}
\begin{proof}
Let $X_1,X_2,X_3,X_4$ be a paw of the set graph such that the non-edges are between $X_4$ and $X_2\cup X_3$. Let $c_j \in \{c_1, \dots, c_{k-1}\}$ be a neighbor of $X_1$ (recall that each $X_i$ is a module \wrt $\{c_1, \dots, c_{k-1}\}$).
We first apply Lemmas~\ref{lem:P3} and \ref{lem:P3bis}, and Lemma~\ref{lemma:clean} for the couples $\{(1, 4), (1, 2)\}$. Hereafter, there are only two cases: either $c_j$ is also adjacent to $X_2$, $X_3$ but not to $X_4$, or is adjacent to $X_4$ but not to $X_2$, $X_3$.\\

\noindent\textit{Case 1:}  $c_j$ is adjacent to $X_2$, $X_3$ but not to $X_4$. Let $u$ be an arbitrary vertex of $X_1$, and assume the instance is positive: let $s_a \in X_a$ be the elements of a rainbow independent set for $a \in \{1, 2, 3\}$. We have the following:
\begin{itemize}
	\item $u$ cannot be adjacent to both $s_2$ and $s_3$ since, in this case, $\{c_j, s_2, s_3, u\}$ together with a neighbor of $u$ in $X_4$ (remember that we applied Lemma~\ref{lemma:clean} for the couple $(1, 4)$, hence such a neighbor must exist) induce a dart.
	\item $u$ cannot be adjacent to $s_1$ and to exactly one vertex among $\{s_2, s_3\}$, since, in this case, $\{u, s_1, s_2, s_3, c_j\}$ induces a dart.
\end{itemize}
 We now create four branches which will correspond to the remaining different possibilities of adjacencies between $u$ and $\{s_1, s_2, s_3\}$: the first three branches (first item below) represent the case where $u$ is adjacent to only one vertex among $\{s_1, s_2, s_3\}$, while the second item represents the case where $u$ is adjacent to none of $\{s_1, s_2, s_3\}$:
 \begin{itemize}
 	\item for every $a \in \{1, 2, 3\}$, create a branch where:
 	\begin{itemize}
 		\item $X_a$ is replaced by $X_a \cap N[u]$
 		\item $X_b$ is replaced by $X_b \setminus N(u)$, for $b \neq a$.
 	\end{itemize}
 	And apply Lemma~\ref{lemma:clean} with a couple $(b, c)$ with $b, c \neq a$ (chosen arbitrarily). 	
	\item in the fourth branch, replace $X_a$ by $X_a \setminus N(u)$ for every $a \in \{1, 2, 3\}$, and apply Lemma~\ref{lemma:clean} with the couple $(2, 3)$.
 \end{itemize}
 Now, observe that:
 \begin{itemize}
 	\item if $u$ is adjacent to only $s_a$, then $\{s_1, s_2, s_3\}$ are still in the reduced graph of the corresponding first item above, and, moreover, there is no edge between $X_b$ and $X_c$, since, because of Lemma~\ref{lemma:clean}, if there was an edge between $X_b$ and $X_c$, we would be able to find an edge $xy$ in $X_b$ and a vertex $z$ in $X_v$ such that $xz$ is an edge but $yz$ is not, and $\{x, y, z, u, c_j\}$ would induce a dart (recall that $u$ is not adjacent to $x$, $y$ and $z$).
 	\item if $u$ is adjacent to none of $\{s_1, s_2, s_3\}$, then these vertices are still in the reduced graph of the second item above, and, moreover, there is no edge between $X_2$ and $X_3$, using similar arguments as previously.
 \end{itemize}
 
\noindent\textit{Case 2:} $c_j$ is adjacent to $X_4$ but not to $X_2$, $X_3$. In this case, we claim that there exists $j' \neq j$ such that $c_{j'}$ is adjacent to $X_1$ or $X_4$ (or both): if this is not the case and the instance is positive, there would be an independent set of size $k$ which intersects $\{c_1, \dots, c_{k-1}\}$ (by considering $\{c_1, \dots, c_{k-1}\} \setminus \{c_j\}$ together with the vertices of the rainbow independent set intersecting $X_1$ and $X_4$), which is impossible in an instance of \faugramsey{$f$}.
However, we may assume that $c_{j'}$ is not adjacent to $X_1$: since we applied Lemma~\ref{lemma:clean} on the couple $(1, 2)$, there exists an edge $xy$ in $X_1$ and a vertex $z$ in $X_2$ such that $xz$ is an edge and $yz$ is not, but in this case $\{x, y, z, c_j, c_{j'}\}$ would induce a dart. Hence $X_1$ must induce an independent set: if it is of size at least $k$ we are done, otherwise we may branch to decide which vertex should be in the solution.
Hence, $c_{j'}$ is adjacent to $X_4$ and not to $X_1$.

\begin{claim} 
We can FPT-reduce to the case where $X_1$ induces a clique.
\end{claim}
\begin{claimproof}
First observe that the neighborhood in $X_1$ of any vertex of $X_4$ must be a clique, for otherwise there would be a dart using $c_j$ and $c_{j'}$. We choose an arbitrary vertex $u \in X_4$. Assume the instance is positive: let $s_1 \in X_1$ and $s_4 \in X_4$ be the elements of a rainbow independent set. Notice that there are three possibilities: either $u = s_4$, or $us_4 \in E(G)$ or $us_4 \notin E(G)$. We perform a branching corresponding to these possibilities: in the first branch $u$ is taken in the solution, and $\kappa$ decreases. In the second branch we remove $N(u)$ from $X_4$, and thus the size of a maximum independent set in $G[X_4]$ decreases, and the same holds for $\kappa$. In the third branch we remove from $X_4$ the non-neighbors of $u$ and perform another branching in order to guess whether $u$ is a neighbor of $s_1$ or not. In the first branch we thus replace $X_1$ by $X_1 \cap N(u)$ which is a clique, as desired. In the second branch we replace $X_1$ by $X_1 \setminus N(u)$, but in this case $X_1$ is now $P_3$-free: indeed, if $a$, $b$, $c$ induces a $P_3$, then $\{a, b, c, u, c_{j}\}$ induces a dart. Hence $X_1$ induces a disjoint union of cliques: if there are at least $k$ of them we are done, and otherwise we branch once again and replace $X_1$ by a clique.
\end{claimproof}

Let us now summarize the situation: $c_j$ is adjacent to $X_1$ and $X_4$ but not to $X_2$ nor $X_3$, and $X_1$ induces a clique.
There must be two vertices $c_{j_2}, c_{j_3} \in \{c_1, \dots, c_{k-1}\}$ ($j_2 \neq j_3$) dominating $X_2 \cup X_3$ (otherwise, we would be able to find an independent set of size $k$ intersecting $\{c_1, \dots, c_{k-1}\}$, as previously).
If one of them, w.l.o.g. $c_{j_2}$, is also adjacent to $X_1$, then we can apply the same argument with $c_{j_2}$ instead of $c_j$. We are then in Case 1 and we can reduce. So we can assume that $c_{j_2}$ and $c_{j_3}$ are not adjacent to $X_1$. 
If both $c_{j_2}$ and $c_{j_3}$ dominate $X_2$, then we claim that we can branch to decrease the number of edges in the set graph. To do so, we claim that if a vertex $x$ of $X_1$ is adjacent to a connected component of $X_2$, it is adjacent to the whole component. Indeed otherwise, there exists an edge $yz$ of $X_2$ such that $xz$ is an edge and not $xy$. And then we obtain a dart with universal vertex $z$ using $c_{j_1},y,c_{j_2}$ and $x$. So we can branch over the connected components of $X_2$ to guess in which connected component is selected the vertex of $X_2$ and then delete the neighbors of this component in $X_1$. In the resulting graph, the edge set of the set graph decreased.

So from now on we assume that $c_{j_2}$ is adjacent to $X_2$ and not to $X_1,X_3$ and $c_{j_3}$ is adjacent to $X_3$ and not to $X_1,X_2$ (and we do not assume anything about the adjacencies of $c_{j_2}$ and $c_{j_3}$ to $X_4$).
We say that an edge between $X_2$ and $X_3$ is \emph{strong} if it is contained in a triangle using a vertex of $X_1$. If $ab$ is a strong edge, we claim that $\{a, b\}$ is a module \wrt $X_1$: otherwise, let $x \in X_1$ such that $abx$ is a triangle, let $x' \in X_1$ such that $ax'$ is an edge but $bx'$ is not: $\{a, b, x, x', c_{j_2}\}$ induces a dart (remember that $xx'$ is an edge, since $X_1$ induces a clique). 
Hence, every connected component $C$ of the subgraph consisting of strong edges is a module \wrt $X_1$. In particular, every such connected component $C$ is in the neighborhood of some vertex $x \in X_1$ which is adjacent to $c_j$ (while $c_j$ is not adjacent to any vertex of $C$). 
Hence, the (connected) subgraph induced by $C$ must be $P_3$-free (otherwise, such a $P_3$ together with $x$ and $c_j$ would induce a dart), it is thus a clique. Both $X_2$ (resp. $X_3$) can therefore be partitioned into $X_2^1$, $\dots$, $X_2^p$ (resp. $X_3^1$, $\dots$, $X_3^p$) such that every vertex of $X_2^a$ is connected with strong edges to every vertex of $X_3^a$, for every $a \in \{1, \dots, p-1\}$, and every vertex of $X_2^p$ (resp. $X_3^p$) is not incident to any strong edge.
We now perform a random reduction : pick at random a non trivial subset $A$ of $[p-1]$, remove from $X_2$ all vertices from $\cup_{a \in A} X_2^a$ and remove from $X_3$ all vertices from $\cup_{a \in [p-1] \setminus A} X_3^a$.
 Assume the instance is positive: let $(s_2, s_3) \in X_2 \times X_3$ be the elements of a rainbow independent set. Since strong edges are edges of $G$, there exist $a_2, a_3 \in [p]$ such that $(s_2, s_3) \in X_2^{a_2} \times X_3^{a_3}$ with $a_2 \neq a_3$ or $a_2=a_3=p$. Hence, the probability that $s_2$ and $s_3$ have not been deleted is at least $1/2$ (and if the instance is negative, the reduced one is negative as well, since the reduced graph is an induced subgraph of the former). 
 Moreover, in the reduced instance, there is no strong edge between $X_2$ and $X_3$ (but, there still might be edges). However, if two vertices $x$, $x'$ in $X_1$ have a common neighbor in $X_2$ (resp. $X_3$), then they must be twins in $X_3$ (resp $X_2$), otherwise we would be able to form a dart together with $c_j$. This means that all edges between $X_1$ and $X_2$ can actually be partitioned into (not necessarily induced) complete bipartite graphs. We can thus perform a random reduction similar to the previous one, except that in the reduced instance, there will not be any edge between $X_1$ and $X_2$. Hence, the parameter $\kappa$ decreases, which concludes the proof.
\end{proof}

From now on, we may assume that the set graph is $\{$claw, paw$\}$-free. In that case, we prove that it has a simple structure: it is either a path, a cycle or the complement of a matching. The first two cases will then be handled in Lemma~\ref{lem:complementmatching} and the last one in Lemma~\ref{lem:pathcycle}.

\begin{lemma}\label{lem:structpaw}
If a connected graph $G$ is $\{$claw, paw$\}$-free, it is either a path, a cycle or the complement of a (not necessarily perfect) matching.
\end{lemma}
\begin{proof}
By contradiction. Let $P$ be a induced path or cycle of maximal size of $G$ (if a path of the same length exists, we choose the path).
First note that if $P \in \{ P_2, C_3 \}$ then $G$ is a clique and the conclusion holds. 
So we can assume that $P$ is $P_k$ with $k \ge 3$ or $C_{k'}$ with $k' \ge 4$.
We prove that all vertices of $V \setminus P$ are either complete or anti-complete to $P$. Indeed let $x$ be a vertex adjacent to a vertex of $P$. If $x$ is adjacent to an internal vertex $y$ of $P$, then since $G$ is claw-free, it must also be adjacent to a neighbor of $y$. If it is not complete to $P$, let $y_1,y_2,y_3$ be consecutive vertices of $P$ such that $x$ is adjacent to both $y_1,y_2$ but not $y_3$. Then $x,y_1,y_2,y_3$ is a induced paw since $P \ne C_3$. If $x$ is not adjacent to the interior of $P$, we can either increase the length of the path or create a cycle with an additional vertex, a contradiction.
Since $G$ is connected, all the vertices that are not adjacent to $P$ have to be adjacent to a vertex dominating $P$, which creates a paw. So such vertices do not exist.

Now remark that if the set of dominating vertices is not empty, then $P$ is either a $P_3$ or a $C_4$ since otherwise there is a paw. Also note that every vertex of $V \setminus P$ has co-degree at most one in $V \setminus P$ since otherwise there is a paw or a claw using a vertex of $P$. So $V \setminus P$ induces the complement of a (not necessarily perfect) matching, and is complete with either $P_3$ or $C_4$, which are themselves the complement of a matching. 
\end{proof}

Because of the previous result and Lemma~\ref{lem:pathcycle}, the only remaining case is where the set graph is the complement of a (non necessarily perfect) matching.

\begin{lemma}[Particular set graph: complement of matching]\label{lem:complementmatching}
If the set graph is the complement of a (not necessarily perfect) matching, then we can output $O(h(k)|V(G)|^2)$ instances in which the set graph is bipartite, for some computable function $h$.
\end{lemma}
\begin{proof}
Assume w.l.o.g. that $X_1$ is a vertex of the set graph of degree at least two in $\{c_1, \dots, c_{k-1}\}$. Such a vertex exists since the bipartite graph between the sets $X_1$ and $\{c_1, \dots, c_{k-1}\}$ is connected. Let $c_j$ and $c_{j'}$ be two neighbors of $X_1$.
We apply Lemma~\ref{lemma:clean} to the set of couples $\{(i, j): i < j\}$, and now claim that the instances we obtain are FPT (without applying any new FPT-reduction).
First observe that for every $r \neq 1$ such that $X_r$ is adjacent to $X_1$, $X_1$ is adjacent to at least one of $\{c_j, c_{j'}\}$. Indeed, otherwise, pick an edge $xy$ in $X_1$ together with a vertex $z$ in $X_r$ such that $xz$ is an edge and $yz$ is not (this configuration is possible since all couples $(1, r)$ are clean for all $r \neq 1$): $\{x, y, z, c_j, c_{j'}\}$ induces a dart.
Hence, $\{c_j, c_{j'}\}$ is adjacent to all the vertices of the set graph but at most one (since $X_1$ has at most one non-neighbor in the set graph).

Assume w.l.o.g. that $X_k$ is the unique (if it exists) non-neighbor of $X_1$. 
We now try all possible choices for $s_1 \in X_1$ and $s_k \in X_k$ (if $X_k$ does not exist, we just try all choices for $s_1$), and for every $\ell \neq 1, k$, we remove from $X_\ell$ all neighbors of $s_1$ and $s_k$.
If some of these sets become empty, then we can answer ``No''. 
Otherwise, observe that every remaining vertex is in the neighborhood of either $c_j$ or $c_{j'}$ and in the non-neighborhood of $s_1$, and $s_1$ is a neighbor of $c_j$ and $c_{j'}$. Hence, the neighborhood of $c_j$ (resp. $c_{j'}$) in the remaining vertices must be $P_3$-free. Thus, all remaining vertices can be partitioned into two disjoint union of cliques $A_1$, $\dots$, $A_p$ and $B_1$, $\dots$, $B_q$. If $p \ge k-2$ or $q \ge k-2$, then we are done ($p \ge k-1$ or $q \ge k-1$ if $X_k$ does not exist). Otherwise, we branch in order to guess which of these cliques contain a solution of a rainbow independent set. The remaining cliques now play the role of the sets $X_{\ell}$, and we thus end up with $O(4^kn^2)$ instances of \faugramsey{$f$} whose set graph is a bipartite graph, as desired.
\end{proof}
Observe that all FPT branchings of our algorithm either end on the case where the set graph is a path or a cycle, in which case we conclude by Lemma~\ref{lem:pathcycle}, or where the set graph is the complement of a matching, in which case we have $O(h(k)n^2)$ instances where the set graph is bipartite. But in the latter case, since every new branching only removes edges of the set graph, it must remain bipartite, which ensures that the branching of Lemma~\ref{lem:complementmatching} will be performed only once, except if the set graph is both the complement of a matching and bipartite, but in this case it must be of size at most $4$, and we conclude by an exhaustive guess instead of applying the above lemma.
\end{proof}

\subsection{The cricket}

\begin{theorem}
There is a randomized FPT algorithm for \mis in cricket-free graphs parameterized by the size of the solution.
\end{theorem}

\begin{proof}
Let us first prove that we can assume some additional information on the graph $G$.

\begin{lemma}\label{lem:dominatingset}
If $G$ has a dominating set of constant size, then we can solve \mis in FPT time.
\end{lemma}
\begin{proof}
Let $D \subseteq V(G)$ be such a dominating set (it can be found in $O(|E(G)| \cdot |V(G)|^{|D|})$ time). 
Let us partition the vertices of $V(G) \setminus D$ into at most $2^{|D|}$ sets $S_1$, $\dots$, $S_p$, depending on their exact neighborhood in $D$. Since $G$ is cricket-free, $G[S_i]$ is $O_4$-free, where $O_4$ is the graph with four vertices and one edge.
We branch in order to guess whether the solution has an intersection with $S_i$ of size $0$, $1$ or at least $2$, for every $i \in [p]$. Observe that there are at most $3^{2^{|D|}}$ such choices. Consider one of these choices.
We delete sets with intersection of size $0$, and guess one vertex in each set with intersection $1$ and decrease $k$ by one. For every set $S_i$ with intersection at least $2$, we guess two vertices and remove their neighborhood. Since $G[S_i]$ is $O_4$-free graphs, the remaining vertices must form an independent set. If one of them is of size at least $k$ we are done. Otherwise, the remaining vertices are of size at most $k \cdot 2^{|D|}$, and we can conclude by brute-force.
\end{proof}

\begin{lemma}\label{lem:cr_k23}
We can reduce to the case where $G$ is $K_{2,3}$-free.
\end{lemma}
\begin{proof}
We prove it by induction on $k$.
Let $C$ be a $K_{2,3}$. If $|N(C)|$ is bounded by a function of $k$, we branch to guess which vertex in $N[C]$ has to be selected in the solution and we decrease the invariant (indeed, for any set of vertices $Q$, any maximal independent set must intersect either $Q$ or $N(Q)$). So from now on, we assume that $N[C]$ is arbitrarily large.
Let us denote by $X$ the independent set of size $2$ in $C$, and by $Y$ the independent set of size $3$. We first consider the set $Z$ of vertices $z$ only adjacent to $X$ but not to $Y$ or to $Y$ but not to $X$. We claim that there are at most $5k$ such vertices: for every $x \in C$, we may assume $|N(x) \cap Z| \ge k$, for otherwise either $N(x) \cap Z$ induces an independent set (in which case we are done), or it contains an edge, in which case this edge together with $x$ and two neighbors of $x$ in $C$ induce a cricket. 

Let $A$ be the set of vertices adjacent to both sides of $C$. We claim that every $a \in A$ is adjacent to at least two vertices of $Y$. Indeed, otherwise $\{a\}$ union $Y$ union a vertex of $X$ adjacent to $a$ induces a cricket. In particular, every vertex adjacent to $C$ is adjacent to an edge and a non-edge of $C$.
Let $a \in A$. All but at most one neighbor of $a$ is adjacent to $C$. Indeed, otherwise $a$ has two neighbors $u,v$ non-adjacent to $C$. If $uv$ is an edge, $u,v,a$ and a non-edge of $N(a) \cap C$ induce a cricket. If $uv$ is a non-edge, then $a,u,v$ and an edge of $N(a) \cap C$ induce a cricket. 

Let $U:=C \cup A$ and $W:=V \setminus (C \cup A)$. Every vertex of $U$ is adjacent to at most $5k+1$ vertices of $W$, namely vertices of $Z$ plus at most one vertex. So $U$ is one-sided almost disconnected to $W$. If $W$ is of size bounded by a function of $k$, notice that $A$ is dominated by two vertices of $Y$, in which case we conclude by Lemma~\ref{lem:dominatingset}. So by induction and by Lemma~\ref{lem:almostDisjoint}, the problem can be decided in FPT time.
\end{proof}

\begin{lemma}\label{lem:cr_k15}
We can reduce to the case where $G$ is $K_{1,5}$-free.
\end{lemma}
\begin{proof}
By Lemma~\ref{lem:cr_k23}, we may assume that $G$ is $K_{2,3}$-free.
Let $C$ be an induced $K_{1,5}$ and let us denote by $v$ the center of the star and $v_1,\ldots,v_5$ its leaves. 

First note that there are at most $k$ vertices $A$ adjacent to $v$ but not adjacent to $\{ v_1,\ldots,v_5\}$. Indeed, otherwise $A$ would contain an edge, and this edge, $v$ and $v_1,v_2$ would induce a cricket.
If a vertex $x$ is adjacent to $v$ and at least one $v_i$, then at most one $v_j$ is not in $N(x)$. Indeed otherwise $x,v,v_i$ and two non-neighbors of $x$ in $\{ v_1,\ldots,v_5 \}$ would form a cricket. Note moreover that no vertex can be adjacent to three leaves of the star but not to $v$ since otherwise there would be a $K_{2,3}$ in $G$, a contradiction with Lemma~\ref{lem:cr_k23}. So $V \setminus C$ can be partitioned into the following sets: $X_i, Y_i, Y_{i,j}, X, Y$ where $1 \le i,j \le 5$. The set $X_i$ denotes the set of vertices whose neighborhood in $C$ is $v$ union all leaves but $v_i$. The set $Y_i$ (resp. $Y_{i,j}$ is the set of vertices whose neighborhood in $C$ is $v_i$ (resp. $\{v_i,v_j \}$). And $X$ (resp. $Y$) is the set of vertices complete (resp. anti-complete to) $C$.

The set $X \cup X_i$ is anti-complete to the sets $Y_{\ell,j}$ for all $i,j$, $Y$ and $Y_j$ for $j \ne i$. Indeed otherwise we create a triangle using this edge plus a common neighbor on $C$ and complete the cricket with two neighbors of the vertex of $X \cup X_i$ that are not neighbors of the vertex of $Y_j$ (resp. $Y_{\ell,j}$). Moreover, every vertex $x \in X \cup X_i$ has at most one neighbor in $Y_i \cup Y$. Indeed, if there is a non-edge in $N(x) \cap (Y_i \cup Y)$, we complete the cricket with the edge $v,v_j$ with $j \ne i$ adjacent to $x$. If there is an edge in $N(x) \cap (Y_i \cup Y)$,  we complete the cricket with the non-edge $v_j,v_j'$ with $j,j' \ne i$ adjacent to $x$.

So every vertex of $U=C \cup X \cup_{i \le 5} X_i$ has degree at most $k+1$ in $W=\bigcup_{i \le 5} Y_i \bigcup_{i,j \le 5} Y_{i,j} \cup Y \cup A$. Note that $U$ is dominated by $v$ and thus, if $|W|$ is bounded by a function of $k$, we can conclude by branching on $W$ and applying Lemma~\ref{lem:dominatingset}. Otherwise, by induction and by Lemma~\ref{lem:almostDisjoint}, the problem can be decided in FPT time.
\end{proof}

In order to prove that \mis is (randomized) FPT in cricket-free graphs, we first apply Lemmas~\ref{lem:cr_k23} and \ref{lem:cr_k15} above. These lemmas ensure that in order to prove our result, it is sufficient to give a (randomized) FPT algorithm for \mis in $\{cricket, K_{2,3}, K_{1,5}\}$-free graphs. To do so, we now prove that \faugramsey{$f$} is (randomized) FPT in $\{cricket, K_{2,3}, K_{1,5}\}$-free graphs, where $f(x) = 2$ for all $x \in \mathbb{N}$. That is, $V(G)$ is partitioned into $X_1 \cup \dots \cup X_k \cup C_1 \cup \dots C_{k-1}$, and, for every $j \in [k-1]$, we have $C_j = \{c_j^1, c_j^2\}$.

In the remainder of the proof, $R$ denotes the set graph.

\begin{lemma}\label{lem:cr_boundeddegree}
 Let $X_iX_j$ be an edge of the set graph. If all the vertices of $X_i$ have degree at most $k$ in $X_j$, we can reduce.
\end{lemma}
\begin{proof}
Assume that all the vertices of $X_i$ have degree at most $k$ on $X_j$. We now use a random reduction. Namely, we delete every vertex of $X_j$ with probability $1/2$. We then delete all vertices $x$ of $X_i$ having a neighbor in $X_j$. After this transformation, there is no edge anymore between $X_i$ and $X_j$ and then the invariant has decreased. 
Moreover, if the instance is positive, then the vertex of $X_j$ in the rainbow solution is still in $X_j$ with probability $1/2$ and the one of $X_i$ is still in $X_i$ with probability at least $1/2^{k}$, which completes the proof.
 \end{proof}

We consider two types of sets $X_i$ depending of the number of elements of $\{c^1_1, \dots, c^1_{k-1}\}$ they see. A vertex $X_i$ of the set graph is \emph{type 1}  if it sees only one element of this set, otherwise it is \emph{type 2}.


\begin{lemma}\label{lem:cr_type2bounded}
In a positive instance, every type 2 vertex of the set graph has degree at most $6$.
\end{lemma}
\begin{proof}
Assume that $X_i$ is of type 2, and let $c^1_j$, $c^1_{j'}$ be adjacent to $X_i$. Since $G$ is $K_{1,5}$-free by Lemma~\ref{lem:cr_k15}, the vertex $c^1_j$ (resp. $c^1_{j'}$) is adjacent to at most three $X_r$'s distinct from $X_i$ (since, if the instance is positive, five elements of a rainbow independent set would induce a $K_{1,5}$). Hence, if $X_i$ has degree at least $7$ in the set graph, it must be adjacent to some $X_r$ which is not itself adjacent to $c^1_j$ nor $c^1_{j'}$. If there is a vertex $x$ in $X_i$ of degree at least $k$ in $X_r$, then either these vertices are independent, in which case we are done, or there is an edge $ab$. But in that case $\{a, b, x, c^1_j, c^1_{j'}\}$ induces a cricket. Hence the degree of every vertex of $X_i$ in $X_r$ is at most $k$, in which case we can reduce by Lemma~\ref{lem:cr_boundeddegree}.
\end{proof}

\begin{lemma}\label{lem:cr_triangle1}
We can reduce so that the set graph does not contain any triangle of type $1$ vertices.
\end{lemma}
\begin{proof}
  W.l.o.g., assume $X_1,X_2,X_3$ is a triangle, each vertex respectively dominated by $C_1,C_2,C_3$. Note that if $C_i$ and $C_j$ with $i \ne j \in \{ 1,2,3\}$ are not distinct then we can replace $c^1_i$ by an independent set of size two in $X_i \cup X_j$ and obtain an independent set of size $k$ intersecting $C_1 \cup \dots \cup C_{k-1}$, which is impossible in an instance of \faugramsey{$f$}. So from now on, we will assume that $C_1,C_2,C_3$ are pairwise distinct.
  
  Now, observe that every vertex $x$ of $X_1$ has a neighborhood in $X_2 \cup X_3$ which is a clique, otherwise there is a cricket. By symmetry the same holds for $X_2$ and $X_3$. So the tripartite graph on $X_1,X_2,X_3$ is a disjoint union of complete tripartite graphs, and we can conclude with a random reduction. Namely, for each component $T$ of the complete tripartite graph between $X_1,X_2,X_3$, we choose to keep $T \cap X_1$ with probability $1/3$, or to keep $T \cap X_2$ with probability $1/3$ or to keep $T \cap X_3$ with probability $1/3$. With probability at least $1/27$, the vertices of the rainbow independent set in $X_1,X_2,X_3$ are still in the resulting graph. Moreover, after this operation, $X_1,X_2,X_3$ is an independent set of the set graph, so the number of edges in the set graph decreased.
\end{proof}

\begin{lemma}\label{lem:cr_claw}
The set graph $R$ is claw-free.
\end{lemma}
\begin{proof}
Assume that a vertex $X_i$ of $R$ has an independent set of size $3$ in its neighborhood, namely $X_1,X_2,X_3$. By Lemma~\ref{lem:cr_boundeddegree}, there exists a vertex $x$ of $X_i$ with at least $k$ neighbors in $X_1$. So there is an edge in its neighborhood in $X_1$. By completing this set with a neighbor of $X_i$ in $X_2$ and $X_3$, we obtain a cricket.
\end{proof}

\begin{lemma}\label{lem:cr_degree}
The maximum degree in the set graph $R$ is  $14$.
\end{lemma}
\begin{proof}
Assume that some $X_i$ has degree at least $15$. By Lemma~\ref{lem:cr_type2bounded}, it must be of type $1$. By Lemma~\ref{lem:cr_triangle1}, the neighborhood of $X_i$ cannot contain two adjacent type 1 vertices. And it cannot contain three non adjacent vertices by Lemma~\ref{lem:cr_claw}. So it has at most $2$ type-1 neighbors.
Now if $X_i$ has $13$ neighbors $\mathcal{X}$ of type 2, since each of them have maximum degree at most $6$ by Lemma~\ref{lem:cr_triangle1} (including type 1 neighbors), the subgraph induced by $\mathcal{X}$ contains an independent set of size $3$, a contradiction with Lemma~\ref{lem:cr_claw}.
\end{proof}

\begin{lemma}\label{lem:cr_bull}
We can reduce so that the set graph $R$ is bull-free.
\end{lemma}
\begin{proof}
  Let $X_1$ be the chin of the bull, and $X_2,X_3$ be the non horns and $X_4,X_5$ be the horns adjacent to respectively $X_2$ and $X_3$. 
  Using an FPT branching, we show that we may assume that for every vertex of $x$ of $X_2$ we have that $N(x) \cap X_4$ contains an edge. Indeed, for a given rainbow independent set, there are two cases: either the vertex of $X_2$ has degree at most $k$ in $X_4$, or it has degree at least $k+1$ in $X_4$. For the first branch, we remove from $X_2$ all vertices with degree greater than $k$, and, by Lemma~\ref{lem:cr_boundeddegree}, we can reduce. We thus end up with the second branch, where every vertex of $X_2$ as degree at least $k+1$ in $X_4$. Since these neighbors are not an independent set (otherwise we are done), they must induce an edge.
  Similarly, we may assume that for every vertex $y$ in $X_3$, we have that $N(y) \cap X_5$ contains an edge.

  First note that, for every edge $xy$ between $X_2,X_3$ $N(x) \cap X_1 = N(y) \cap X_1$. Indeed otherwise, we can assume by symmetry that there exists $z \in X_1$ adjacent to $x$ but not to $y$. Then, using an edge in the neighborhood of $x$ in the horn, we make a cricket.
  
 Using the same idea as previously, we may also assume that every vertex of $X_2$ has degree at least $k$ in $X_3$. 
  
  We distinguish two cases, which correspond to two branches. First assume that the bipartite graph $B$ between $X_2$ and $X_3$ (\ie not taking into account the edges induced by $X_2$ nor $X_3$) is connected. Then all the vertices of $X_2,X_3$ have exactly the same neighborhood in $X_1$. Since we are looking for an independent set containing a vertex of $X_1,X_2$ and $X_3$, all the neighbors of vertices of $X_2$ in $X_1$ can be deleted. After this modification, there is no edge anymore between $X_1$ and $X_2$ nor between $X_1$ and $X_3$.
  
  So we may assume that the bipartite graph $B$ between $X_2$ and $X_3$ is not connected.
  If the solution does not select a vertex in $X_2$ and in $X_3$ in the same component of $B$, then we can conclude using a random reduction: for each connected component $T$ of the bipartite graph, we keep $T \cap X_2$ with probability $1/2$ or we keep $T \cap X_3$ with probability $1/2$. The vertices of the rainbow independent set in $X_2$ and $X_3$ are still in the graph with probability $1/4$. After this branching, the number of edges in the set graph decreases.

  So we may finally assume that the solution selects a vertex in $X_2$ and $X_3$ in the same connected component of $B$. Remark that, for every edge $x,x' \in X_2$ where $x,x'$ lie in distinct connected components of $B$, then $x,x'$ have the same neighborhood in $X_4$. Indeed, otherwise we can assume w.l.o.g. that there is a vertex $y \in X_4$ adjacent to $x$ but not to $x'$. Then $x,x',y$ plus an edge of $N(x) \cap X_3$ induces a cricket. 
  
  Let us denote by $H$ the subgraph of $G[X_2]$ where $xy$ is an edge if $xy$ is an edge of $G$ and $x,y$ are not in the same component of $B$. We now run the following algorithm: we start with $S=\emptyset$ and $T=\emptyset$  and $W=\emptyset$. As long as there remains a component $B'$ of $B$ and $C'$ of $H$ such that $B' \notin S$ and $C' \notin T$ and there exists $x \in B' \cap C'$, we add $B$ in $S$ and $C$ in $T$ and $x \in W$. We repeat this operation as long as we can. We claim that $W$ is an independent set of $G$. Indeed assume by contradiction that $xx'$ is an edge and that $x'$ is added in $W$ after $x$. The vertex $x'$ is not in the connected component of $x$ in $B$ by definition of $S$; and it is not in the connected component of $x$ in $H$ by definition of $T$. Since every edge of $G[X_2]$ that is not in $H$ is in $X_2 \cap B'$ for some component $B'$ of $B$, we have a contradiction.
  So finally all the vertices of $X_2$ are in the connected component of $S$ in $B$ or in the connected component $T$ in $H$. We branch over all the possible choices in order to guess in which connected component of $S$ and $T$ the vertex of $X_2$ in a rainbow solution lies (in each branch, we replace $X_2$ by the corresponding connected component). Clearly there at most $2k$ choices. In the resulting branchings for $S$, all the vertices of $X_2$ now have the same neighborhood in $X_1$ and then we can reduce the number of edges in the set graph. And in the resulting branchings for $T$, all the vertices of $X_2$ now have the same neighborhood in $X_4$ and then we can reduce the number of edges in the set graph, which completes the proof.
%
%
%
%
%
\end{proof}

\begin{lemma}
  We can solve the reduced instance in polynomial time.
\end{lemma}
\begin{proof}
The idea is to prove that the pathwidth of the set graph is bounded by some constant, and then apply a dynamic programming similar to the ones of Lemma~\ref{lem:pathcycle} where the set graph is a path.

  Consider a longest induced path $P$ in the set graph $R$. Since the maximum degree of the set graph is bounded by Lemma~\ref{lem:cr_degree} and that $R$ is connected (and since we can directly conclude if $R$ is small enough), we can assume that $P$ has length at least $14$.  
  
  Note that $P$ dominates $R$. Indeed let $X$ be a vertex of $R$ not adjacent to $P$. Since $R$ is connected, we can assume that $R$ is at distance two from $P$ and let $Y$ be a neighbor of $X$ adjacent to $P$. If $Y$ is only adjacent to an endpoint of $P$, $P$ is not maximal. If it is adjacent to both endpoints but not the internal vertices, there is claw, a contradiction with Lemma~\ref{lem:cr_claw}. If it is adjacent to an internal vertex (but not its neighbors), there is a claw, a contradiction with Lemma~\ref{lem:cr_claw}. If it is adjacent to some (but not all) vertices of $P$, then there is a bull, a contradiction with Lemma~\ref{lem:cr_bull}. So $P$ dominates $R$.
  
  Let $Y$ be a vertex in the neighborhood of $P$. We claim that either $Y$ sees an endpoint of $P$ or that its neighborhood in $P$ consists of at most $4$ consecutive vertices. Indeed, assume that $Y$ is not connected to an endpoint of $P$ and let $X_i$ be its rightmost neighbor on the path $P$. If $Y$ is not adjacent to $X_{i-1}$, the vertex before $X_i$ in $P$, there is a claw, a contradiction with Lemma~\ref{lem:cr_claw}. If it is adjacent to $X_{i-1}$ but not $X_{i-2}$, there is a bull, a contradiction with Lemma~\ref{lem:cr_bull}. So $Y$ has to be adjacent to all of $X_i,X_{i-1},X_{i-2}$. But then if it is adjacent a vertex $X_j$ with $j \le i-4$, then there is a claw, a contradiction with Lemma~\ref{lem:cr_claw}.
  
  It implies that the set graph $R$ has pathwidth at most $85$. Indeed, let us consider a path $p_1,\ldots,p_\ell$ of length $\ell:=|P|$ . Every vertex of $R$ adjacent to an endpoint of $P$ is added in all the bags. For any other vertex $X_i \notin P$ whose first neighbor on $P$ is $p_j$, we add $X_i$ in the bags of $p_j,\ldots,p_{j+3}$. We finally add the vertex $X_i$ of $P$ in the $p_i$ an $p_{i-1}$. All bags contain at most $6 \cdot 14+2= 86$ sets $X_i$ by Lemma~\ref{lem:cr_degree}. So the pathwidth of the set graph is at most $85$. Note moreover that this decomposition can be found in polynomial time (we start with a maximal path by inclusion and either we improve it or we find the decomposition).
  
  So we can now find an independent set in polynomial time using dynamic programming. We order the sets $X_i$ in such a way that if $X_j \succ X_i$ the vertex $X_j$ first appear in a bag later than $X_i$. We now claim that choosing a vertex in $X_{i+1}$ knowing $X_1,\ldots,X_i$ is equivalent to choosing a vertex $X_{i+1}$ only knowing $(\cup_{j \le i} X_i) \cap \mathcal{B}$ where $\mathcal{B}$ is the first bag of the path decomposition containing $X_{i+1}$. Indeed, by definition of path decomposition, if $X_j$ with $j \le i$ is not in $\mathcal{B}$, then there is no edge in the set graph between $X_iX_j$ (since $X_j$ only appears in a subpath of $P$ and has already disappeared). 
  So in order to decide which further vertices can be selected, we only need to keep track of the vertices selected in the current bag $\mathcal{B}$ of the path decomposition. Since the size of the bags is bounded, we obtain a polynomial time dynamic programming algorithm to decide the problem in that case.
\end{proof}
That finishes the proof.
\end{proof}




\end{document}